\documentclass[11pt,twoside,letter,reqno]{amsart}
\pdfoutput=1
\usepackage{graphicx}
\usepackage{fancyhdr}
\usepackage{subfig}
\usepackage{cite}
\newtheorem{definition}{Definition}
\usepackage[table]{xcolor}
\newtheorem{theorem}{Theorem}
\newtheorem{lemma}{Lemma}
\newtheorem{corollary}{Corollary}

\usepackage{mathtools}
\mathtoolsset{showonlyrefs}

\makeatletter
\let\uppercasenonmath\@gobble
\let\scshape\relax
\makeatother

\makeatletter
\def\specialsection{\@startsection{section}{1}%
  \z@{\linespacing\@plus\linespacing}{.5\linespacing}%
  {\normalfont}}
\def\section{\@startsection{section}{1}%
  \z@{.7\linespacing\@plus\linespacing}{.5\linespacing}%
  {\normalfont\scshape\bfseries}} 
\makeatother

\usepackage{scalerel}
\usepackage{stackengine}

\begin{document}

\newcommand{\R}{\mathcal{R}_N}
\newcommand{\SM}{\mathcal{S}_M}
\newcommand{\dist}{\operatorname{dist}}
\newcommand{\RM}{\mathcal{R}_M}

\stackMath
\def\hatgap{2pt}
\def\subdown{-2pt}
\newcommand\reallywidehat[2][]{%
	\renewcommand\stackalignment{l}%
	\stackon[\hatgap]{#2}{%
		\stretchto{%
			\scalerel*[\widthof{$#2$}]{\kern-.6pt\bigwedge\kern-.6pt}%
			{\rule[-\textheight/2]{1ex}{\textheight}}
		}{0.5ex}
		_{\smash{\belowbaseline[\subdown]{\scriptscriptstyle#1}}}%
}}

\newcommand{\wal}{\operatorname{wal}}
\newcommand{\sal}{\operatorname{sal}}
\newcommand{\calo}{\operatorname{cal}}
\newcommand{\Wal}{\operatorname{Wal}}
\newcommand{\Sal}{\operatorname{Sal}}
\newcommand{\Cal}{\operatorname{Cal}}

\title{\bf\vspace{-39pt} Linear reconstructions and the analysis of the stable sampling rate}

\author{Laura Thesing \\ \small Department of Applied Mathematics and Theoretical Physics, \\ University of Cambridge, Wilberforce Road, Cambridge \\ \small lt420@cam.ac.uk \\
\\
Anders Hansen \\ \small Department of Applied Mathematics and Theoretical Physics,\\ University of Cambridge, Wilberforce Road, Cambridge \\ \small ach70@cam.ac.uk}

\date{}

\begin{abstract}
	The theory of sampling and the reconstruction of data has a wide range of applications and a rich collection of techniques. For many methods a core problem is the estimation of the number of samples needed in order to secure a stable and accurate reconstruction. This can often be controlled by the {\it Stable Sampling Rate (SSR)}. In this paper we discuss the SSR and how it is crucial for two key linear methods in sampling theory: generalized sampling and the recently developed Parametrized Background Data Weak (PBDW) method. Both of these approaches rely on estimates of the SSR in order to be accurate. In many areas of signal and image processing binary samples are crucial and such samples, which can be modelled by Walsh functions, are the core of our analysis.  As we show, the SSR is linear when considering binary sampling with Walsh functions and wavelet reconstruction. Moreover, for certain wavelets it is possible to determine the SSR exactly, allowing sharp estimates for the performance of the methods.
	\vspace{5mm}  
	\vspace{3mm}
\end{abstract}

\maketitle
\thispagestyle{fancy}

\markboth{\footnotesize \rm \hfill L. THESING AND A. HANSEN \hfill}
{\footnotesize \rm \hfill STABLE SAMPLING RATE \hfill}

\section{Introduction}

Sampling theory is a core principle in image and signal processing as well as the mathematics of information, data science and inverse problems. Since the early results of Shannon \cite{ShannonOverview, Shannon, Shannon50} many techniques have been developed, and there are now a myriad of methods available. Moreover, the many applications, such as Magnetic Resonance Imaging (MRI)  \cite{MRI, Unser_MRI}, electron tomography \cite{lawrence2012et,leary2013etcs}, lensless cameras, fluorescence microscopy \cite{Candes_PNAS, Roman}, X-ray computed tomography \cite{Stanford_CT, quinto2006xrayradon}, surface scattering \cite{Nature_sci_rep} as well as parametrised PDEs \cite{deVore1, deVore3, deVore2}, make the field well connected to different areas of the sciences.
 
A standard sampling model is as follows. 
We have an element $f \in \mathcal{H}$, where $\mathcal{H}$ is a separable Hilbert space, and the goal is to reconstruct an approximation to $f$ from a finite number of linear samples of the form $l_i(f)$, $i \in \mathbb{N}$. In particular, given that the $l_i$s are linear functionals, we measure the scalar product between $f$ and some sampling element $s_i \in \mathcal{H}$, $i\in \mathbb{N}$, i.e. $l_i(f) = \langle f, s_i \rangle$. 
It is important to note that the  $l_i$s cannot be chosen freely, but are dictated by the modality of the sampling device, say an MRI scanner providing Fourier samples or a fluorescence microscope giving binary measurements modelled by Walsh coefficients. A natural question that arises from this setting is, what is the number of samples that is needed for a accurate and stable reconstruction? This can be made explicit by the stable sampling rate (SSR).

To define the SSR we first introduce the sampling space and the reconstruction space. We define the  {\it sampling space} $\mathcal{S} = \overline{\operatorname{span}} \{ s_i : i \in \mathbb{N} \} \subset \mathcal{H}$, meaning the closure of the span. In practice, one can only acquire a finite number of samples, therefore, we denote by $\mathcal{S}_M = \operatorname{span} \{ s_i : i = 1, \ldots, M \}$ the sampling space of the first $M$ elements. Similarly, the  reconstruction space denoted by $\mathcal{R}$ is spanned by reconstruction functions $(r_i)_{i \in \mathbb{N}}$, i.e. $\mathcal{R} = \overline{\operatorname{span}} \{ r_i : i \in \mathbb{N} \}$. As in the sampling case, one has to restrict to a finite reconstruction space, which is denoted by $\mathcal{R}_N = \operatorname{span} \{ r_i : i = 1 , \ldots , N \}$.

The key ingredient in the definition of the SSR is the  \emph{subspace angle} $\omega$ between the subspaces $\R$ and $\SM$. In particular, 
	\begin{equation}
	\cos(\omega(\R,\SM)) := \inf_{r \in \R, \|r\|=1} \|P_{\SM} r \|.
	\end{equation} 
The orthogonal projection onto the sampling space is denoted by $P_{\SM}$. Mainly, one is interested in the reciprocal value 
\[
\mu(\R,\SM) =  1/\cos(\omega(\R,\SM)) \in [1, \infty],
\]
which, as we will see below, plays a key role in all the error estimates of the two linear algorithms discussed here. Due to the definition of cosine, $\mu$ takes values in $[1,\infty]$.
We can now define the \emph{stable sampling rate}
\begin{equation}
\Theta(N, \theta) = \min \left\{ M \in \mathbb{N}: \mu(\mathcal{R}_N,\mathcal{S}_M)  < \theta \right\}.
\end{equation}
In particular, the SSR determines how many samples $M$ are needed when given $N$ reconstruction vectors in order to bound $\mu$ by $\theta$.

Although the SSR is developed mainly for linear methods there is a strong connection to non-linear techniques. For example in infinite-dimensional compressed sensing a key condition is the so-called {\it balancing property} \cite{GSinfCS}, which is very similar to the SSR. When the balancing property is satisfied, the truncated matrix $P_{[N]}UP_{[M]}$ is close to an isometry. Hence, it also ensures the stability of the reconstruction.

The stable sampling rate has already been analysed for different settings. A prominent one is the reconstruction from Fourier measurements. In this case the sampling functions $s_i$ are the complex exponentials. For the reconstruction of wavelets it was shown in \cite{linearity} that the SSR is linear. The other main measurements are binary. Combined with wavelets they also have a linear SSR \cite{WalshSSR}.
 
The purpose of this paper is twofold:
	\begin{itemize}
		\item First, we want to compare two linear reconstruction methods: generalized sampling \cite{GS} and the PBDW approach based on data assimilation \cite{deVore3}. We also show how both linear methods completely rely on the SSR in order to be accurate. Additionally, the non-linear extension of generalized sampling, which is presented in \cite{Clarice}, is included in the comparison.
		\item Second, we provide sharp results on the SSR when considering Walsh functions and Haar wavelets. This can be done by realising the common structure of the Walsh functions and the Haar wavelets. Although the SSR is linear when considering Walsh samples and wavelet reconstructions, sharp results on the SSR for arbitrary Daubechies wavelets are still open. The difficulty is that the higher order Daubechies wavelets share very little structural similarities with the Walsh functions. 
	\end{itemize}

\section{Reconstruction Methods}

In terms of reconstruction methods, there are three different properties that are often desired. The most important are obviously \emph{accuracy} and \emph{stability}. However, \emph{consistency}, meaning that the reconstruction will yield the same samples as the true solution, is also often considered an advantage. Below, we will see how the SSR is crucial for the two former properties.

\subsection{Reconstruction and Sampling Space}

Throughout the paper $\mathcal{H} = L^2([0,1]^d)$. Due to the fact that we are dealing with the $d$-dimensional case, we introduce multi-indices to make the notation more readable. Let $j = \left\{j_1, \ldots, j_d \right\} \in \mathbb{N}^d$, $d \in \mathbb{N}$, be a multi-index. A natural number $n$ is in the context with a multi-index interpreted as a multi-index with the same entry, i.e. $n = \left\{n, \ldots ,n \right\}$. Then, we define the addition of two multi-indices for $j, r \in \mathbb{N}^d$ by the point-wise addition, i.e. $j + r = \left\{j_1 +r_1, \ldots, j_d +r_d \right\}$ and the sum $\sum_{j=k}^{r} x_j := \sum_{j_1 = k_1}^{r_1} \ldots \sum_{j_d = k_d}^{r_d} x_{j_1, \ldots, j_d}$, where $k,r \in \mathbb{N}^d$. The multiplication of a multi-index with a real number is understood point-wise as well as the division by a multi-index.

We now discuss the sampling space $\mathcal{S}_M$. As pointed out before, we are interested in binary measurements. Binary measurements come from applications such as fluorescence microscopy or single pixel cameras. Mathematically they are modelled by scalarproducts of our function of interest $f$ and sampling functions $s_i$, $i \in \mathbb{N}$, which take only the values $\left\{0,1\right\}$ or $\left\{-1,1\right\}$. Without loss of generality, we assume in the further scope of this paper that the functions take values $\left\{-1,1\right\}$. The former can be attained by taking an additional measurement with the constant function. For practical reasons it is desirable to have a fast transform for these measurements. Therefore, measurement matrices which consist of random matrices are not desirable. Additionally, it has been shown that the reconstruction quality is better with more structured sampling spaces. Therefore, we use Walsh functions to represent the measurements. Walsh functions obey the advantage that they correspond to a fast transform. Moreover, they are the kernels of the Hadamard transform and the analogue or dyadic version of exponential functions. Hence, they are a natural choice to model binary measurements. Now, we define the Walsh functions.
The Walsh functions in higher dimensions can be represented by the tensor product of one-dimensional Walsh functions. 

\begin{definition}[Walsh function \cite{deutschWalsh}]
	Let $n =\sum_{i\in \mathbb{Z}} n_i 2^{i-1}$ with $n_i \in \left\{0,1 \right\}$ be the dyadic expansion of $s \in \mathbb{R}_+$. Analogously, let $x = \sum_{i \in \mathbb{Z}} x_i 2^{i-1}$ with $x_i \in \left\{ 0,1 \right\}$. The generalized Walsh functions in $L^2([0,1])$ are given by 
	\begin{equation}
	\operatorname{Wal}(n,x) = (-1)^{\sum_{i \in \mathbb{Z}} (n_i + n_{i+1})x_{-i-1}}.
	\end{equation}
	We extend it to functions in $L^2([0,1]^d)$ by the tensor product for $n = (n_k)_{k=1,\ldots,d}, x = (x_k)_{k=1,\ldots,d}$
	\begin{equation}
	\operatorname{Wal}(n,x) = \bigotimes_{k=1}^d  \operatorname{Wal}(n_k,x_k) .
	\end{equation}
\end{definition}
Notice that the parameter $n$ represents the number of zero crossings of the function. For this reason, it is often referred to as sequency, which is similar to frequency in the case of exponential functions.
The Walsh functions span the sampling space, i.e. for $M=m^d, m\in \mathbb{N}$ we have
\begin{equation}
\mathcal{S}_M = \operatorname{span} \left\{ \Wal(n,\cdot), n = (n_k)_{k=1, \ldots,d}, n_k =1,\ldots,m, k = 1, \ldots, d \right\}.
\end{equation}
Moreover, Walsh functions can be extended to negative inputs by $\Wal(-n,x) = \Wal(n,-x) = -\Wal(n,x)$.

With the help of the Walsh functions we can define the continuous Walsh transform of a function $f \in L^2([0,1]^d)$ as in \cite{deutschWalsh} almost everywhere by
\begin{equation}
\reallywidehat[W]{f}(n) = \langle f(\cdot), \Wal(n,\cdot) \rangle = \int_{[0,1]^d} f(x) \Wal(n,x)dx, ~ n \in \mathbb{R}^d.
\end{equation}
The Walsh functions have the following useful properties. First, they obey the scaling property, i.e. $\Wal(2^jn,x) = \Wal(n,2^jx)$ for all $j \in \mathbb{N}$ and $n, x \in \mathbb{R}$. Second, the multiplicative identity holds, this means $\Wal(n,x)\Wal(n,y) = \Wal(n,x \oplus y)$, where $\oplus$ is the dyadic addition, i.e. the element-wise addition modulo two. These properties are also easily transferred to the Walsh transform. For further information on Walsh functions and transforms see \cite{WalshAndTheirAppl, dyadicAnalysis, dyadicDerivative}.

Direct inversion from a finite number of samples, both in the Fourier and Walsh case, may lead to substantial artefacts such as the Gibbs phenomenon in the Fourier case or block artefacts known from Walsh functions. This can be seen in the numerical experiments in Figures \ref{fig:HaarFourTF64}, \ref{fig:WalHelpTF128} and \ref{fig:TF2D128}. Therefore, it is important to consider reconstruction spaces $\mathcal{R}$  which represent the data in a better way so that a finite and low-dimensional subspace leads to a good reconstruction. In many different applications, such as image and signal processing and also representations of solution manifolds for PDEs \cite{OrderReductionWav}, wavelets have become highly popular alternatives. 
The reconstruction space is spanned by reconstruction functions $r_i$, $i \in \mathbb{N}$. As it is not possible to deal numerically with an infinite number of samples, it also is not possible to reconstruct infinite number of coefficients. Therefore, we examine the reconstruction space $\R = \operatorname{span}\left\{ r_i: i =1,\ldots,N \right\}$. When $\R$ is used as an approximation for the solution manifold of a PDE, we are also given the approximation error $\epsilon_N$.

In the following, we use wavelets as the reconstruction space due to their good time and frequency localisation. First, we study the one-dimensional case. Then, we get to higher dimensions. We use the common notation and denote the mother wavelet with $\psi$ and the corresponding scaling function with $\phi$. These functions are then scaled and translated. This results in the functions
\begin{equation}
	\psi_{R,j}(x) := 2^{R/2} \psi(2^R x - j) \text{ and } \phi_{R,j}(x) := 2^{R/2} \phi(2^R x - j),
\end{equation}
where $R,j \in \mathbb{Z}$. The wavelet space at a certain level $r$ is given by $W_r := \operatorname{span} \left\{\psi_{r,j} : j \in \mathbb{Z} \right\}$ and the scaling space is given by $V_r := \operatorname{span} \left\{\phi_{r,j} : j \in \mathbb{Z} \right\}$. Often one is interested in the representation of functions in $L^2([0,1])$ instead of $L^2(\mathbb{R})$. For this sake boundary corrected wavelets were introduced in \cite{boundaryWavelets}. We are focussing on the version presented in chapter $4$ of \cite{boundaryWavelets}, because they still obey the same smoothness and vanishing moments properties as the wavelets they are derived from on the real line. Additionally, they also remain the multi-resolution analysis property. We now shortly repeat the construction. The scaling space for boundary corrected wavelets is spanned by the original scaling function $\phi$ and reflections around $1$ of the scaling function $\phi^\#$, i.e.
\begin{equation}
	V_r^b = \operatorname{span}\left\{ \phi_{r,j} : j = 0,\ldots,2^r -p -1 , \phi_{r,j}^\# : j = 2^r - p,\ldots,2^r-1  \right\},
\end{equation}
for Daubechies wavelets of order $p$.
The boundary corrected Daubechies wavelets still obey the multi-resolution analysis. Therefore, it is possible to represent the union of the wavelet spaces up to a certain level $R-1$ by the scaling space at this level, i.e. 
\begin{equation}
	\bigcup_{r < R} W_r^b = V_R^b.
\end{equation}
Hence, it is not necessary for the analysis to have a deeper look in the ordering of the wavelets and their construction for the boundary corrected version as long as we have that the number of coefficients equals the number of elements in that level, i.e. if $N = 2^R$ for some $R \in \mathbb{N}$. The reconstruction space is then given by
\begin{equation}
	\R : = V_R^b.
\end{equation}
The higher-dimensional scaling spaces are constructed by the tensor product of the one-dimensional one, such that we get in $d$ dimensions
\begin{equation}
	\R = V_R^{b,d} := V_R^b \otimes \ldots \otimes V_R^b \quad \text{ (d-times)}
\end{equation}
for $N = 2^{dR}$. Note that the corresponding wavelet space is not simply the tensor product of the one-dimensional wavelets. Fortunately, this is not often a problem in the analysis, since Daubechies wavelets obey the multi-resolution analysis. We will deal with the internal ordering for the one- and two-dimensional case for the Haar wavelets.

\subsection{Reconstruction Techniques}\label{CH:ReconstructionTechniques}

In the following, we present two different interesting reconstruction methods, which are both optimal in their setting. We highlight some advantages and disadvantages. Particularly, we see that the performance of both methods depends highly on the subspace angle between the sampling and the reconstruction space. This gives rise to the discussion in \S \ref{CH:Linearity} about the question for which sampling and reconstruction spaces the stable sampling rate is linear.

\subsubsection{PBDW-method}

In \cite{deVore1, deVore3, deVore2} the PBDW-method from \cite{PBDW} is analysed. In contrast to the prominent application for image and signal analysis this method arises from the application with PDEs. One tries to estimate a state $f$ of a physical system by solving a parametric family of PDEs
\begin{equation}
	\mathcal{F}(f,\zeta)=0,
\end{equation}
where $\mathcal{F}$ is a differential operator. The parameter $\zeta$ may not be known exactly. Therefore, the value of $f$ cannot be attained by simply solving the PDE. Hence, other information is necessary. In most applications, one has access to linear measurements $l_i(f)$, $i =1,\ldots,M$ of the state $f$. This alone is not sufficient to estimate $f$ or more ambitiously even the parameter $\zeta$. Fortunately, one also has information about the PDE, which can be used to analyse the solution manifold $\mathcal{M}$. The solution manifold is usually quite complicated and not given directly. Hence, approximations are used. The method used in this context is the approximation by a sequence of nested finite subspaces
\begin{equation}
\mathcal{R}_0 \subset \mathcal{R}_1 \subset \ldots \subset \R , \quad \dim (\mathcal{R}_j) = j,
\end{equation} 
where the approximation error is known to be $\epsilon_j$ for each subspace $\mathcal{R}_j$. There are a lot of different methods that allow us to construct the spaces $\mathcal{R}_j$, such as the \emph{reduced basis method} \cite{binev2011convergence, buffa2012priori, devore2013greedy, wojtaszczyk2015greedy} and the use of wavelets \cite{OrderReductionWav}.

The given setting leads to the goal of trying to merge the data driven and model based information, which leads to the concept of \emph{PBDW-method}.
Let the measurement $s := P_{\SM}f$ be given. The idea is to combine the information given by the measurements and the PDE. 
This means we are searching for an approximation $f^* \in \mathcal{K}_s$ where
\begin{equation}
	\mathcal{K} = \left\{ f \in \mathcal{H}: \operatorname{dist}(f,\R) \leq \epsilon_N \right\} \text{ and } \mathcal{H}_s = \left\{ f \in \mathcal{H}: P_{\SM} f = s \right\}.
\end{equation}
The intersection is then the space of possible solutions, i.e. $\mathcal{K} \cap \mathcal{H}_s = \mathcal{K}_s$. The aim is to reduce the distance between the approximation $f^*$ and the true solution $f$.

It was shown in \cite{deVore3} that the following approach introduced in \cite{PBDW} is optimal for this task. First, the minimization problem
\begin{equation}\label{Eq:DataAv*}
	g^* = \operatorname{argmin}_{g \in \R} || s - P_{\SM} g||^2,
\end{equation}
is solved. The solution to the reconstruction problem is then given by the mapping $A^*: \SM \rightarrow \mathcal{H}$ defined by
\begin{equation}
	A^*(P_{\SM}f) := f^* = s + P_{\SM^\perp} g^*.
\end{equation}

For the performance analysis, we want to compare the reconstruction algorithm $A^*$ with all other mappings  $A:\SM \rightarrow \mathcal H$. In the following, let $A$ represent any alternative reconstruction algorithm which also only depends on the information of $f$ in the sampling space $\SM$.

We first examine the \emph{instance optimality}. This means we analyse the error for a given measurement $s$, i.e. 
\begin{equation}
||f - A(s)|| \leq C_A(s)\operatorname{dist}(f,\R), \quad f \in \mathcal{K}_s.
\end{equation}
The algorithm $A$ which leads to the smallest constant $C_A(s)$ is called \emph{instance optimal}. It is clear that the error scales with the distance of the element $f$ from the reconstruction space $\R$. Due to the fact that we normally do not know $s$ a priori, this estimate is not very helpful. Hence, one is interested in the performance for any kind of input $s \in \mathcal{S}_M$. Therefore, the \emph{performance of a recovery algorithm $A$} on any subset $W \subset \mathcal{H}$ is given by
\begin{equation}
E_A(W) := \sup _{f \in W} || f - A(P_{\mathcal{S}_M}f) ||.
\end{equation}

Taking the infimum over the whole class gives the \emph{class optimal performance} on a given set $W$ defined by
\begin{equation}
E(W) := \inf_{A} E_A(W).
\end{equation}

It is shown in \cite{deVore3} that the presented algorithm is both instance and class optimal and gives the estimate
\begin{equation}\label{Eq:ErrorboundPBDW1}
||f - A^*(P_{\mathcal{S}_M}f)|| \leq \mu(\mathcal{R}_N,\mathcal{S}_M)\operatorname{dist}(f,\mathcal{R}_N).
\end{equation}
Hence, with this approach we do not get reasonable estimates, if $\mathcal{R}_N \cap \mathcal{S}_M^\perp \neq \left\{0\right\}$. Moreover, a detailed knowledge about the stable sampling rate is necessary to get a method which can be used in practice. A reconstruction method is of little help if we do not have good error bounds or if the condition number is too high. We will see in the next chapter that the condition number of this reconstruction method is also bounded if the stable sampling rate has been taken into account, see Equation \eqref{Eq:Kappa}. Next, it was shown in \cite{PBDW} that this estimate can even be improved to 
\begin{equation}\label{eq:est} 
||f - A^*(P_{\SM}f)|| \leq \mu(\R,\SM)\dist(f,\R \oplus (\SM \cap \R^\perp)).
\end{equation}
Note that \eqref{eq:est} demonstrates how the PBDW-method is dependent on the SSR. Moreover, it was shown in \cite{deVore3} that the constant $\mu(\R,\SM)$ cannot be improved. Thus, the estimate is sharp, which further demonstrates the importance of the SSR.

\subsubsection{Generalized Sampling}

After the examination of the concept of the PBDW-method, we want to study the different reconstruction technique \emph{generalized sampling}. Here, the approach is a stable improvement of concepts as finite section methods \cite{FiniteSection1, FiniteSectionGroechnig, FiniteSection3, FiniteSection4}. The main difference is that generalized sampling allows for different dimensions on $\SM$ and $\R$, whereas in the finite section method they are always the same. However, the finite section method becomes a special case of generalized sampling when the dimensions of the sampling space and reconstruction space are equal. The method is defined now. Afterwards, we discuss how the equation can be attained by solving the equivalent least square problem.

\begin{definition}[\cite{2DCase}]
	For $f \in \mathcal{H}$ and $N,M \in \mathbb{N}$ we define the reconstruction method of \emph{generalized sampling} $G_{N,M} : \mathcal{H} \rightarrow \mathcal{R}_N$ by
	\begin{equation}\label{EqGeneralSamp}
	\langle P_{\mathcal{S}_M} G_{N,M}(f), r_i \rangle = \langle P_{\mathcal{S}_M}f, r_i \rangle, \quad i = 1,\ldots,N .
	\end{equation}
	We also refer to $G_{N,M}(f)$ as the \emph{generalized sampling reconstruction} of $f$.
\end{definition}

We stress at this point that generalized sampling is also a linear reconstruction scheme. In particular, Equation \eqref{EqGeneralSamp} is equivalent to solving the following linear equation for $\alpha^{[N]} \in \mathbb{R}^N$. 
\begin{align}
U^{[N,M]} \alpha^{[N]} & = l(f) ^{[M]} ,
\end{align}
where
\begin{align}\label{Eq:DefU}
U^{[N,M]} & = \begin{pmatrix}
u_{11} & \ldots & u_{1N} \\
\vdots & \ddots & \vdots \\
u_{M1} & \ldots & u_{MN}
\end{pmatrix}  
\end{align}

and $u_{ij} = \langle  r_j, s_i \rangle ,~ l(f)^{[M]} = (l_1(f), \ldots, l_M(f)) \in \mathbb{R}^M$. The matrix can be seen in Figure \ref{fig:RecMatrix} for different sampling and reconstruction spaces. The reconstruction is given by $G_{N,M}(f) = \sum_{i=1}^{N} \alpha_i r_i$. An interesting fact about this is that the matrix $U^{[N,N]}$ is very ill-conditioned for most cases. As pointed out in \cite{linearity}, in the case of Fourier samples and wavelet reconstructions, the condition number grows exponentially in $N$.
This means that this approach is only feasible due to the above mentioned allowance of a different number of samples and reconstructed coefficients. It can be shown that there exists a certain number of samples such that \eqref{EqGeneralSamp} obeys a solution. 

\begin{theorem}[\cite{sharpBounds}]
	Let $N \in \mathbb{N}$. Then, there exists $M_0 \in \mathbb{N}$ such that for every $f \in \mathcal{H}$ Equation \eqref{EqGeneralSamp} has a unique solution $G_{N,M}(f)$ for all $M \geq M_0$. Moreover, the smallest $M_0$ is the least number such that $\cos(\omega(\mathcal{R}_N,\mathcal{S}_{M_0})) >0$.
\end{theorem}

Just as for the PBDW-method, the performance bounds of generalized sampling were studied. Here we observe again that the reconstruction quality highly depends on the subspace angle and on the relation between the data and the reconstruction space.

\begin{theorem}[\cite{sharpBounds}]\label{Th:accuracyGS}
	Retaining the definitions and notations from this chapter, for all $f \in \mathcal{H}$ we have 
	\begin{equation}
	|| G_{N,M}(f)|| \leq \mu(\mathcal{R}_N,\mathcal{S}_M) ||f||,
	\end{equation}
	and 
	\begin{equation}
	|| f - P_{\mathcal{R}_N} f|| \leq || f - G_{N,M}(f)|| \leq \mu(\mathcal{R}_N,\mathcal{S}_M) || f - P_{\mathcal{R}_N} f|| .
	\end{equation}
	In particular, these bounds are sharp.
\end{theorem}

Additionally, it was shown in \cite{GS} that the condition number $\kappa(\mathcal R_N, \mathcal S_M)$ of the reconstruction method can also be controlled by the SSR, i.e.
\begin{equation}\label{Eq:Kappa}
	\kappa(\mathcal R_N, \mathcal S_M) =\mu(\mathcal R_N, \mathcal S_M).
\end{equation}
Hence, not only the accuracy but also the stability is controlled by the SSR.

We conclude that along with mappings that map into the reconstruction space $\R$, generalized sampling is also optimal in terms of achieving a low condition number and high accuracy.

\subsubsection{Non-linear extension of generalized sampling}

In \cite{Clarice} a consistent approach to generalized sampling was introduced. This means that, in contrast to generalized sampling, the output of the reconstruction method has the same measurements in the sampling space as the input. Let $f \in \mathcal{R}$, then it can be represented as $f = \sum_{i=1}^\infty \beta_i r_i$. The measurements can be written as $l(f)^{[M]} = P_{[M]} U \beta$, where $P_{[M]}$ is the orthogonal projection onto space spanned by the first $M$ elements of the canonical bases of $\ell^2(\mathbb N)$, and $U$ is defined as in \eqref{Eq:DefU}. The introduced method solves the non-linear minimization problem
\begin{equation}\label{Eq:infGSdef}
	\inf_{\alpha \in \ell_1} || \alpha ||_{\ell_1} \text{ with } P_{[M]} U \alpha = P_{[M]} U \beta.
\end{equation}
The solution is given by
\begin{equation}
	EG_M(f) = \sum_{i=1}^\infty \alpha_i r_i.
\end{equation}
The measurements of $EG_M(f)$ and $f$ are naturally equal, because of $P_{\mathcal S_M} EG_M(f) = P_{[M]} U \alpha = P_{[M]} U \beta = P_{\mathcal{S}_M} f$. Hence, this approach is consistent and maps into the reconstruction space. In contrast to generalized sampling, it is not necessary to decide the number of reconstructed coefficients a priori.

In the setting of arbitrary sampling and reconstruction spaces, it was shown that for every number of reconstructed coefficients $N$ there exists some $M_0$ such that the method reconstructs the data correct up to its first $N$ coefficients. Hence, this approach is convergent as the error is given by $\mathcal{O}(||P_{\R}^\perp f||)$, which goes to zero for $M \rightarrow \infty$. The speed of convergence is then dependent upon the sampling and reconstruction space. This speed was analysed in \cite{Clarice} for the case of Fourier measurements and wavelet reconstruction. Let the reconstruction space $\mathcal{R}$ be given by wavelets and the sampling space $\mathcal{S}$ be the space representing Fourier measurements. Then, for $\beta \in \ell_1(\mathbb{N})$ and $N \in \mathbb{N}$ the following holds:
	\begin{enumerate}
	\item If for some $D >0$ and $\gamma \geq 1$, the Fourier transform of the scaling function $\phi$ decays with
	\begin{equation}
		|\hat{\phi}(\xi)| \leq \frac{D}{(1+|\xi|)^\gamma}, \quad \xi \in \mathbb{R}
	\end{equation}
	then there exists some constant $C$ independent of $N$ (but dependent on $\gamma$ and $\epsilon$) such that for $M = C N^{1+1/(2\gamma -1)}$, any solution $\alpha$ to \eqref{Eq:infGSdef} satisfies
	\begin{equation}
		|| \alpha - \beta ||_{\ell_1} \leq 6 || P_{[N]}^\perp \beta ||_{\ell_1}.
	\end{equation}
	\item If for $k =0,1,2$, for some $D >0$ and $\gamma \geq 1.5$, the Fourier transform of the scaling function $\phi$, the wavelet $\psi$ and their first two derivatives decays with
	\begin{equation}
		| \hat{\phi}^{(k)} (\xi)| \leq \frac{D}{(1+|\xi|)^\gamma},~ | \hat{\psi}^{(k)} (\xi)| \leq \frac{D}{(1+|\xi|)^\gamma}, \quad \xi \in \mathbb{R},
	\end{equation}
	then there exists some constant $C$ independent of $N$ (but dependent on $\phi$, $\psi$ and $\epsilon$) such that for $M = CN$, any solution $\alpha$ to \eqref{Eq:infGSdef} satisfies
	\begin{equation}
		||\alpha - \beta ||_{\ell_1} \leq 6 || P_{[N]}^\perp \beta ||_{\ell_1}.
	\end{equation}
	\end{enumerate}

For Dauchechies wavelets among other wavelets, the assumptions of this result are satisfied by their construction. The first assumption is fulfilled by all Daubechies wavelets and the second one is fulfilled for Daubechies wavelets of $7$ or more vanishing moments. Numerical experiments suggest that even Daubechies wavelets with less vanishing moments might have a linear relationship \cite{Clarice}.

\subsubsection{Comparison}\label{Ch:Comparison}

In this chapter we want to point out the similarity of generalized sampling and the PBDW-method in terms of optimality and the dependence on the SSR. Additionally, we want to discuss the differences concerning the output space and the consistency. We also include the non-linear extension of generalized sampling into this comparison and examine the overall robustness.
	
In Equation \eqref{eq:est} and Theorem \ref{Th:accuracyGS} we see that the error bounds of generalized sampling and the PBDW-method depend both on $\mu(\R,\SM)$. Moreover, since both systems solve the same least square problem 
\begin{equation}
	U^{[N,M]} \alpha^{[N]} = l(f)^{[M]}
\end{equation}
as part of the reconstruction, they also have the same condition number $\kappa(\R,\SM)$. In more detail, we have that the PBDW-method first solves the generalized sampling problem with a solution $v^*$, see Equation \eqref{Eq:DataAv*}. Then, the solution is tweaked to be consistent, i.e.
\begin{equation}
f^* = g^* + P_{\SM} f - P_{\SM} g^*.
\end{equation}
Moreover, both methods have been proven to be optimal in their setting. This means that linear methods cannot outperform them. Additionally, the bounds are sharp. This underlines the importance of the analysis of the stable sampling rate. Therefore, we dedicate \S \ref{CH:SSR} to the investigation of the relation between the dimension of the sampling and the reconstruction space to bound the subspace angle.

Besides the similarities both methods take different parts into consideration for the construction of the algorithms. Generalized sampling ensures that the output is contained in the reconstruction space. This is very desirable when the signal is sparse in the reconstruction bases, because we do not get the artefacts from the measurements in the output signal. Nevertheless, the approach is not consistent. This means that the projection onto the sampling space of the input and output may not be equal.

This is overcome with the non-linear approach presented in \cite{Clarice}. For this gain in consistency one, unfortunately, has to compromise with the robustness. Stability is considered in the $\ell_1$ setting, and the used definition is equal to the existence of the condition number of the method. This means that the problem is well-conditioned in terms of solving an $\ell_1$ minimization problem rather than the reconstruction of $f$ from $M$ samples. Particularly, this means that the problem is not robust against noise.

Finally, the PBDW method is linear and consistent. This is possible due to the fact that the solution is not forced to stay in the reconstruction space $\R$. Remark that \eqref{eq:est} shows that the error gets smaller with larger $M$ even though we keep the reconstruction space $\R$ the same. Hence, with increasing $\SM$ we are leaving the reconstruction space $\R$ and get further away as $M$ increases. This can be desirable, if the artefacts from the sampling space are not too severe, as it allows us a mix of properties from the sampling and the reconstruction space. Nevertheless, if the function is very sparse in the reconstruction space and has a lot of artefacts in the sampling space, this approach leads to less impressive reconstructions. The impact is further illustrated in \S \ref{Ch:Numerics} Numerical experiments.

\section{Stable Sampling Rate}\label{CH:SSR}

\subsection{Linearity of the Stable Sampling Rate}\label{CH:Linearity}

In \S \ref{CH:ReconstructionTechniques} we saw that the subspace angle between the sampling and the reconstruction space controls the reconstruction accuracy. Therefore, one is interested in the relation between the number of samples and the number of reconstructed coefficients, such that the subspace angle is bounded. In detail, we are interested in the stable sampling rate:
\begin{equation}\label{ss_rate}
\Theta(N, \theta) = \min \left\{ M \in \mathbb{N}: \mu(\mathcal{R}_N,\mathcal{S}_M)  < \theta \right\}.
\end{equation}
The stable sampling rate has been analysed for important cases which appear frequently in practice, i.e. for the Fourier-Wavelet \cite{linearity}, Fourier-Polynomial bases \cite{hrycakIPRM} and Walsh-Wavelet cases \cite{WalshSSR}. For the reconstruction with wavelets we get for both Fourier and Walsh sampling that the stable sampling rate is linear. This is the best relation one can wish for and means that the methods discussed above are, up to a constant, as good as if one could access the wavelet coefficients directly. In particular, for the Walsh case we have the following theorem.

\begin{theorem}[\cite{WalshSSR}]\label{Th:SSR}
	Let $\mathcal{S}$ and $\mathcal{R}$ be the sampling and reconstruction space spanned by the d-dimensional Walsh functions and separable boundary wavelets respectively. Moreover, let $N = 2^{dR}$ with $R \in \mathbb{N}$. Then for all $\theta \in (1, \infty)$ there exists $S_\theta$ such that for all $M \geq 2^{dR}S_\theta$ we have $\mu(\mathcal{R}_N,\mathcal{S}_M) \leq \theta$. In particular one gets $\Theta \leq S_\theta N$. Hence, the relation $\Theta(N;\theta) = \mathcal{O}(N)$ holds for all $\theta \in (1,\infty)$.
\end{theorem}

A natural question which arises from this theorem is whether it is possible to give sharp bounds on the constant $S_\theta$. In Figure \ref{fig:RecMatrix} we can see the stable sampling rate for different Wavelets and the bound $\theta =2$. The slope $S_\theta$ is unknown in most cases and very difficult to find. This comes from the fact that for the majority of wavelets the reconstruction matrix is not perfectly block diagonal, as in \ref{fig:RecMatrixDB2D1} and \ref{fig:RecMatrixDB8D1}. Hence, one has to take the off diagonals into consideration. The numerics suggest that the slope is higher the further away the reconstruction matrix gets from block diagonal. Only for the case of Haar wavelets and Walsh functions we get that the reconstruction matrix is perfectly block diagonal. This can be seen in \ref{fig:RecHaarWalsh}. Note, that from the numerical example one may deduce that $S_\theta =1$. This is indeed the case, and the analysis detailed below establishes that $S_\theta =1$ for all $\theta \in (1,\infty). $

\begin{figure}
	\centering
	\subfloat[Haar Wavelet where $\Theta(N,2) = N$]{\includegraphics[width=0.45\textwidth]{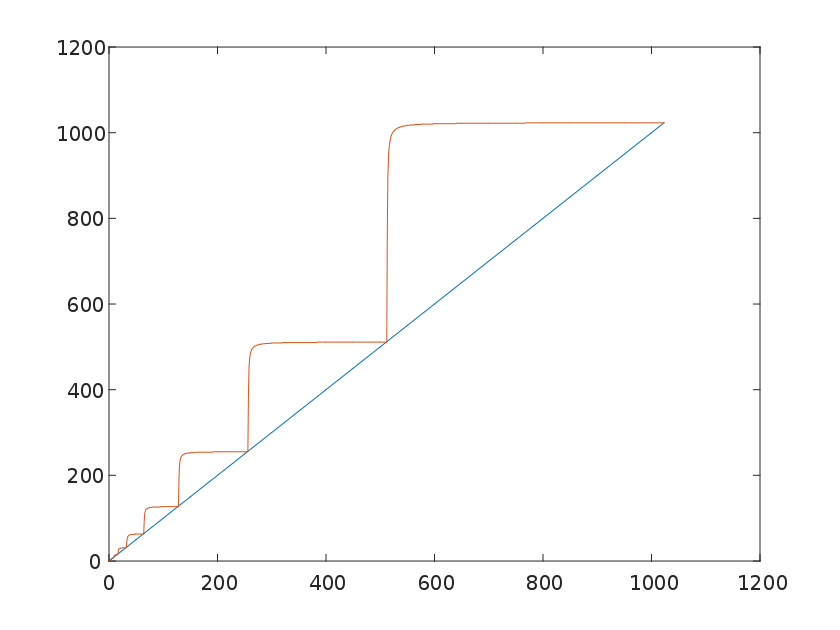}}\quad
	\subfloat[Haar-Walsh]{\includegraphics[width=0.45\textwidth]{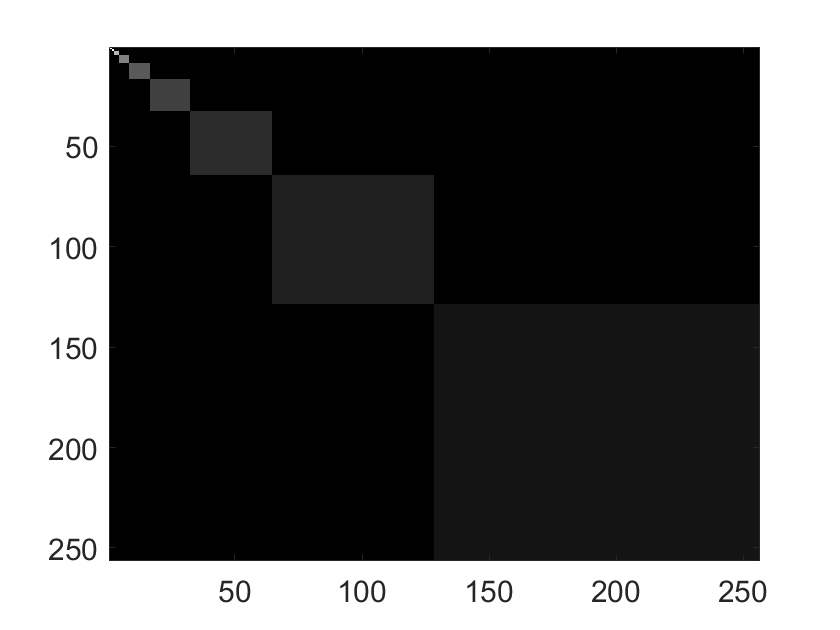}\label{fig:RecHaarWalsh}} \\
	\subfloat[Daubechies $2$ Wavelet where $\Theta(N,2) = 1.49N$]{\includegraphics[width=0.45\textwidth]{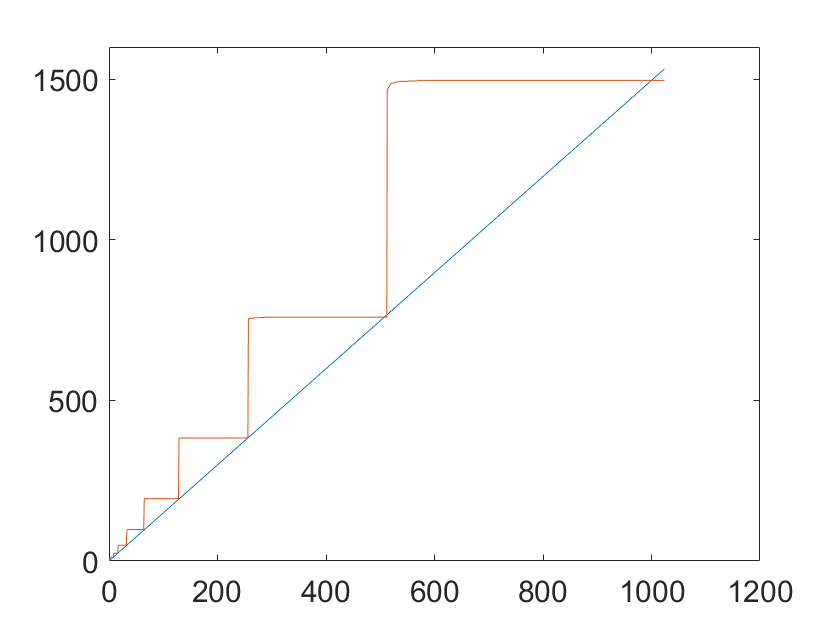}} \quad
	\subfloat[db$2$ - Walsh]{\includegraphics[width=0.45\textwidth]{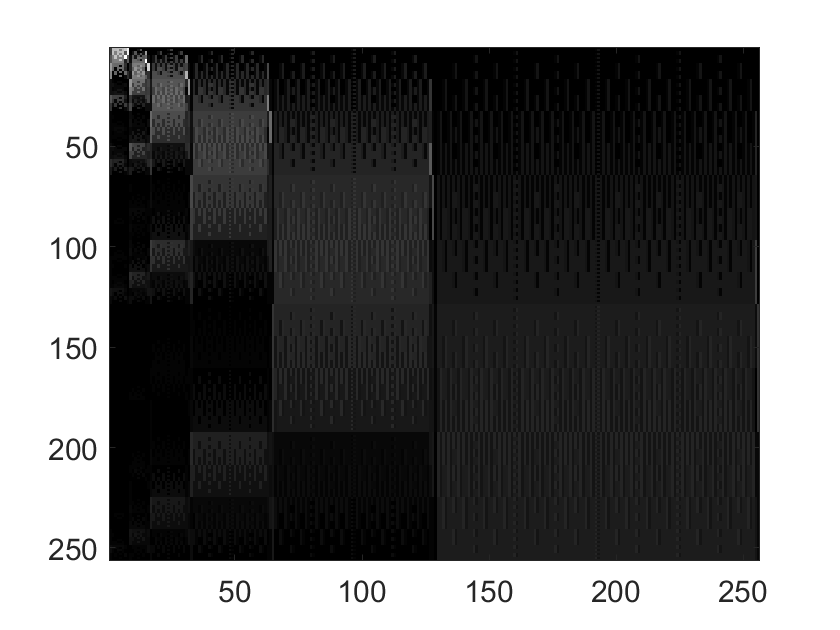}\label{fig:RecMatrixDB2D1}} \\
	\subfloat[Daubechies $8$ Wavelet where $\Theta(N,2) = 2N$]{\includegraphics[width=0.45\textwidth]{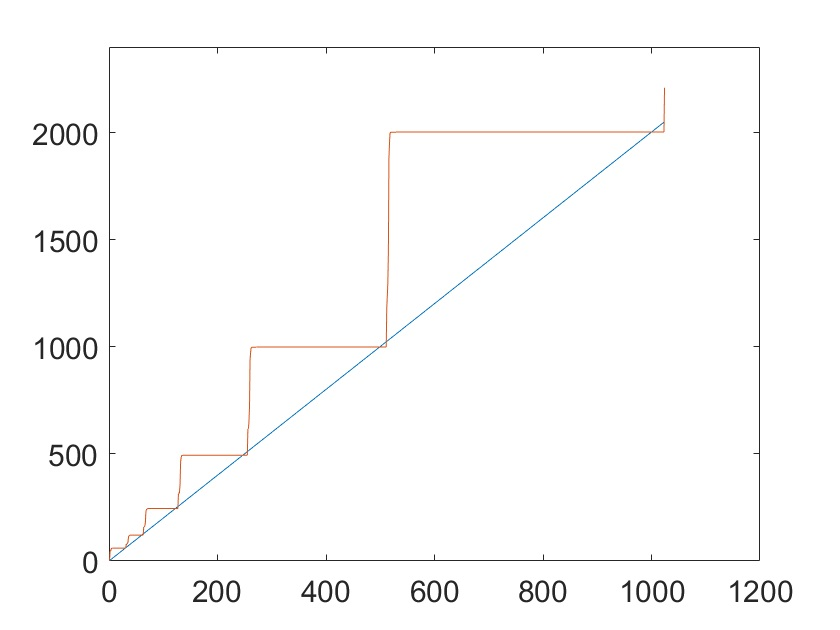}} \quad
	\subfloat[db$8$ - Walsh ]{\includegraphics[width=0.45\textwidth]{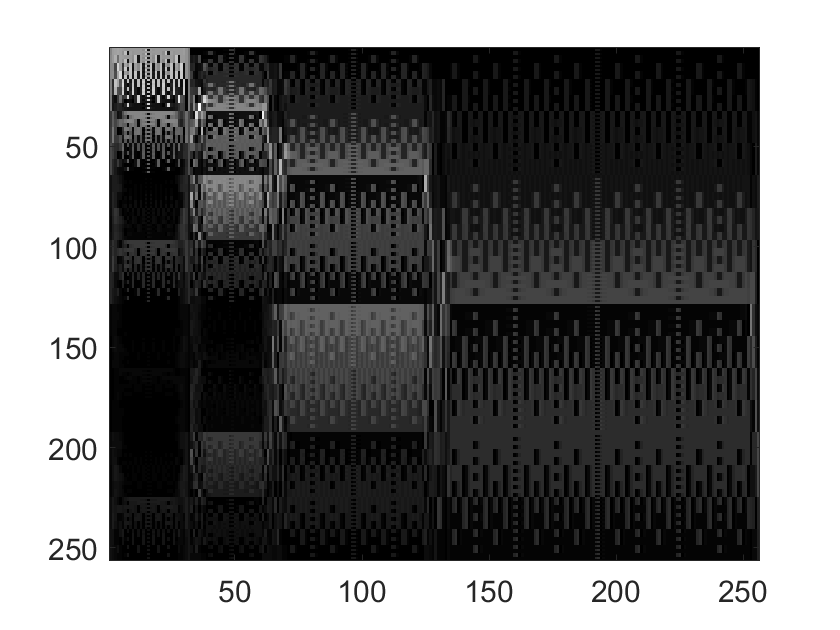}\label{fig:RecMatrixDB8D1}}
	\caption{Stable sampling rate for $\theta =2$ and reconstruction matrix}\label{fig:RecMatrix}
\end{figure}

\subsection{Sharpness for the Haar wavelet - Walsh case}\label{Ch:HaarWalsh}

The sharp bound on $S_\theta$ can be summarised in the following theorem.

\begin{theorem}\label{Th:SSRHaar}
Let the sampling space $\mathcal{S}$ be spanned by the Walsh functions and the reconstruction space $\mathcal{R}$ by the Haar wavelets in $L^2([0,1]^d)$. If $N = 2^{dR}$ for some $R \in \mathbb{N}$, then for every $\theta \in (1,\infty)$ we have that the stable sampling rate is the identity, i.e. $\Theta(N, \theta) = N$. 
\end{theorem}

For the proof we first study the behaviour of Haar wavelets in one dimension under the Walsh transform. This way we also get a theoretical argument for the block structure that can be seen in the numerical implementation.

\begin{lemma}\label{Th:HaarWavDec}
	Let $\psi = \mathcal{X}_{[0,1/2]} - \mathcal{X}_{(1/2,1]}$ be the Haar wavelet. Then, we have that 
	\begin{equation}
		|\langle \psi_{R,j}, \Wal(n,\cdot) \rangle| = \begin{cases}
		2^{-R/2} &\quad 2^R \leq n < 2^{R+1}, 0 \leq j \leq 2^R -1 \\
		 0 &\quad  \text{otherwise.}
		\end{cases}
	\end{equation}
\end{lemma}

\begin{proof}
	For the scalarproduct we have
	\begin{align}
		\langle \psi_{R,j}, \Wal(n,\cdot) \rangle 
		& = \int_0^1 2^{R/2}\left(\mathcal{X}_{[0,1/2]}(2^R x - j) - \mathcal{X}_{(1/2,1]}(2^R x - j) \right)\Wal(n,x) dx \\
		& = 2^{R/2} \left( \int_{\Delta_{2j}^{R+1}} \Wal(n,x) dx  - \int_{\Delta_{2j+1}^{R+1}} \Wal(n,x) dx \right),
	\end{align}
	where $\Delta_k^p = [2^{-p}k, 2^{-p}(k+1))$. We know from \cite{Vegard} that the function $\Wal(n,x)$ for $2^p \leq n <2^{p+1}$ takes the value $+1$ on the interval $\Delta_{2k}^{p+1}$ or $\Delta_{2k+1}^{p+1}$ and $-1$ on the other one for $k =0,\ldots,2^{p}-1$. 
	Now, we consider three different cases. \\[10pt]
	{\bf Case 1: $n < 2^R$}.
	There exist $r < R$ such that $2^{r} \leq n <2^{r+1}$. Then, the function $\Wal(n,x)$ is constant on the interval $\Delta^{r}_k$ for any $k = 0, \ldots, 2^{r}-1$. Note that for $j = 0, \ldots, 2^R -1$ we have that the rounding error is bounded as follows
	\begin{equation}
		\lfloor 2^{r-R} j \rfloor \geq 2^{r-R}j - ( 1 - 2^{r-R}).
	\end{equation}
	Then, we have the interval inclusion 
	\begin{align}
			\Delta^{R+1}_{2j} & = [2^{-R-1}(2j), 2^{-R-1}(2j+1)) \\ 
			&= [2^{-r}2^{-R-1+r}2j, 2^{-r}2^{-R-1+r}(2j+1)) \\
			& \subset [2^{-r} \lfloor 2^{r-R}j \rfloor, 2^{-r}(\lfloor 2^{r-R}j \rfloor+1 )) =  \Delta^{r}_{\lfloor 2^{r-R}j \rfloor}.
	\end{align}
	Moreover, 
	\begin{equation}
		\Delta^{R+1}_{2j+1} \subset \Delta^{r}_{\lfloor 2^{r-R}j \rfloor}.
	\end{equation}
	Hence, $\Wal(n,x)$ takes the same value on $\Delta^{R+1}_{2j}$ and $\Delta^{R+1}_{2j+1}$. Therefore, the two integrals are equal and the scalarproduct vanishes. \\[10pt]
	{\bf Case 2:} $2^R \leq n < 2^{R+1}$.
	We have for $j = 0, \ldots, 2^R -1$ that $\Wal(n,x)$ is equal to $+1$ on $\Delta_{2j}^{R+1}$ or $\Delta_{2j+1}^{R+1}$ and $-1$ on the other. Therefore, we have either
	\begin{align}
		\langle \psi_{R,j}, \Wal(n,x) \rangle & = 2^{R/2} \left( \int_{\Delta_{2j}^{R+1}} \Wal(n,x) dx  - \int_{\Delta_{2j+1}^{R+1}} \Wal(n,x) dx \right) \\
		& = 2^{R/2} \left( \int_{\Delta_{2j}^{R+1}} 1 dx - \int_{\Delta_{2j+1}^{R+1}} -1 dx  \right) = 2^{R/2}  (2^{-R}) = 2^{-R/2}.
	\end{align}
	Or in the other case analogously
	\begin{align}
	\langle \psi_{R,j}, \Wal(n,x) \rangle &  = 2^{R/2} \left( \int_{\Delta_{2j}^{R+1}} -1 dx - \int_{\Delta_{2j+1}^{R+1}} 1 dx  \right) = 2^{R/2}  (-2^{-R}) = -2^{-R/2}.
	\end{align}
	Now, we are left with the last case.\\[10pt]
	{\bf Case 3:} $n \geq 2^{R+1}$.
	There exists an integer $r \geq R+1$ such that $2^r \leq n < 2^{r+1}$. This is similar to the first case. Moreover, we have for $j = 0, \ldots, 2^R -1$ that
	\begin{align}
		\Delta_{2j}^{R+1} = \bigcup_{l=2^{r-R}j}^{2^{r-R}j + 2^{r-R-1} -1 }\Delta_{2l}^{r+1} \cup \Delta_{2l+1}^{r+1}.
	\end{align}
	With the fact that $\Wal(n,x)$ takes the value $+1$ on one of the invervals $\Delta^r_{2l}$ and $\Delta^r_{2l+1}$ and $-1$ on the other, we have that the function $\Wal(n,x)$ takes the values $+1$ and $-1$ on half of the interval of $\Delta_{2j}^{R+1}$. Therefore, the integral vanishes. The same holds true for $\Delta_{2j+1}^{R+1}$ such that we get the desired result.
	
\end{proof}

Before we prove Theorem \ref{Th:SSRHaar}, we analyse the reconstruction matrix for the Haar wavelet - Walsh case in two dimensions. The number of functions which span the wavelet space in $d$ dimensions grows exponentially with $d$. Therefore, we restrict ourselves to two dimensions to underline the main idea. In Figure \ref{fig:RecMatrix2d} we can see that the reconstruction matrix in two dimensions has an additional structure in each level. Similar to the one-dimensional case we have perfect block structure for the Haar case and nearly block structure for the higher order wavelets. For the analysis of this phenomena in the Haar wavelet - Walsh case we examine the definition and order of the two dimensional Haar wavelet. Note that this is necessary for the analysis of the reconstruction matrix, instead of the SSR. In two dimensions the wavelets are constructed by the tensor product. Due to the multi-resolution analysis we have that $V_r^2 = (V_{r-1}\oplus W_{r-1})^2 = V_{r-1}^2 \oplus V_{r-1} \otimes W_{r-1} \oplus W_{r-1} \otimes V_{r-1} \oplus W_{r-1}^2$. Hence, the wavelets can be in three different spaces, $V_{r-1} \otimes W_{r-1}$, $W_{r-1} \otimes V_{r-1}$ and $W_{r-1}^2$. In the first space wavelets are constructed by the tensor product of the one-dimensional scaling function and the one-dimensional wavelet. For the second one the one-dimensional wavelet and the one-dimensional scaling function are combined by the tensor product. And finally, for the third space we take the tensor products of two one-dimensional wavelets. This results for $R \in \mathbb{N}, 0 \leq j_1, j_2 \leq 2^R -1$ in
\begin{equation}\label{Eq:HaarWavelet2D}
	\psi_{R,j_1,j_2,l}(x_1, x_2) = \begin{cases}
	\phi_{R,j_1}(x_1) \psi_{R,j_2}(x_2) \quad & l =1 \\
	\psi_{R,j_1}(x_1) \phi_{R,j_2}(x_2) \quad & l =2 \\
	\psi_{R,j_1}(x_1) \psi_{R,j_2} (x_2)\quad & l =3.
	\end{cases}
\end{equation}
The scaling function is simply the tensor product of two one-dimensional scaling functions:
\begin{equation}
	\phi(x_1,x_2) = \phi(x_1) \phi(x_2).
\end{equation}
For the order of the reconstruction matrix, we first take the first level scaling function $\phi$. Then, we increase by levels, in each level $R$ we let first $j_1$ go from $0, \ldots, 2^{R}-1$ and then $j_2 = 0, \ldots, 2^R-1$. Finally, we let $l =1,\ldots,3$ such that we get for the order of the wavelets: $\phi$, $\psi_{R,0,0,1}, \ldots, \psi_{R,2^R-1,0,1}$,$ \psi_{R,0,1,1}, \ldots$ $ \psi_{R,2^R-1,1,1},$ $\ldots \psi_{R,2^R-1,2^R-1,1}$, $\psi_{R,0,0,2},\ldots, \psi_{R,2^R-1,2^R-1,3}$.

Due to the fact that the higher dimensional wavelets are constructed also by means of the scaling functions, it is necessary to analyse the decay rate for the scaling function as well. Moreover, this is also a main ingredient for the proof of Theorem \ref{Th:SSRHaar} as we represent the union of the wavelet spaces by the scaling space.

\begin{lemma}\label{Lem:HaarScalDec}
	Let $\phi = \mathcal{X}_{[0,1]}$ be the Haar scaling function. Then, we have that the Walsh transform obeys the following block and decay structure
	\begin{equation}\label{Eq:scalingDecay}
		|\langle \phi_{R,j}, \Wal(n,\cdot) \rangle| = \begin{cases}
			2^{-R/2} \quad & n < 2^R, 0 \leq j \leq 2^R -1 \\
			0 \quad &\text{otherwise.}
		\end{cases} 
	\end{equation}
\end{lemma}

\begin{proof}
	The scalarproduct can be expressed as an integral over the interval $\Delta_{j}^R$.
	\begin{align}
		\langle \phi_{R,j}, \Wal(n,\cdot) \rangle & = \int_0^1 2^{R/2} \mathcal{X}_{[0,1]}(2^R x - j) \Wal(n,x) dx \\
		& = 2^{R/2} \int_{\Delta_j^R} \Wal(n,x) dx.
	\end{align} 
	We look at the two different cases \\
	\textbf{Case 1:} $n <2^{R} \quad$ Remember from Lemma \ref{Th:HaarWavDec} that $\Wal(n,x)$ is constant to $+1$ or $-1$ on the interval $\Delta_j^R$ for $j = 0, \ldots, 2^R -1$. Hence, we get that
	\begin{equation}
|\langle \phi_{R,j}, \Wal(n,\cdot) \rangle|  = | 2^{R/2} \int_{\Delta_j^R} \Wal(n,x) dx| = 2^{-R/2}.
	\end{equation}
	\textbf{Case 2:} $n \geq 2^R \quad$ This follows as in Case 3 of Lemma \ref{Th:HaarWavDec}. With the difference that we are looking at the integral over the interval $\Delta_j^R$ instead of the two integrals $\Delta_{2j}^{R+1}$ and $\Delta_{2j+1}^{R+1}$. Nevertheless, they vanish for the same reason.
\end{proof}

With this in hand we can now state the structure of the reconstruction matrix in two dimensions.

\begin{corollary}
	Let $\psi_{R,j_1,j_2,l}$ be the Haar wavelet defined as in \eqref{Eq:HaarWavelet2D}. Then, the Walsh transform has the following property for $0 \leq j_1, j_2 \leq 2^R -1$
	\begin{equation}
	|\langle \psi_{R,j_1,j_2,1} , \Wal(n_1,n_2,\cdot,\cdot) \rangle | = \begin{cases}
	 2^{-R} \quad & n_1 \leq 2^R, 2^R \leq n_2 < 2^{R+1} \\
	 0 & \text{otherwise} ,
	\end{cases}
	\end{equation}	
	\begin{equation}
		|\langle \psi_{R,j_1,j_2,2} , \Wal(n_1,n_2,\cdot,\cdot) \rangle | = \begin{cases}
		 2^{-R} \quad & 2^{R} \leq n_1 < 2^{R+1}, n_2 \leq 2^{R} \\
		 0 & \text{otherwise} 
		\end{cases}
	\end{equation}
	and for the third version
	\begin{equation}
		|\langle \psi_{R,j_1,j_2,3} , \Wal(n_1,n_2,\cdot,\cdot) \rangle | = \begin{cases}
			 2^{-R} \quad & 2^{R} \leq n_1 < 2^{R+1}, 2^R \leq n < 2^{R+1} \\
			 0 & \text{otherwise.} 
		\end{cases}
	\end{equation}
\end{corollary}

\begin{proof}
	The proof follows directly from the tensor product structure and Theorem \ref{Th:HaarWavDec} and Lemma \ref{Lem:HaarScalDec}.
\end{proof}

\begin{figure}
	\centering
	\subfloat[Haar-Walsh]{\includegraphics[width=0.31\textwidth]{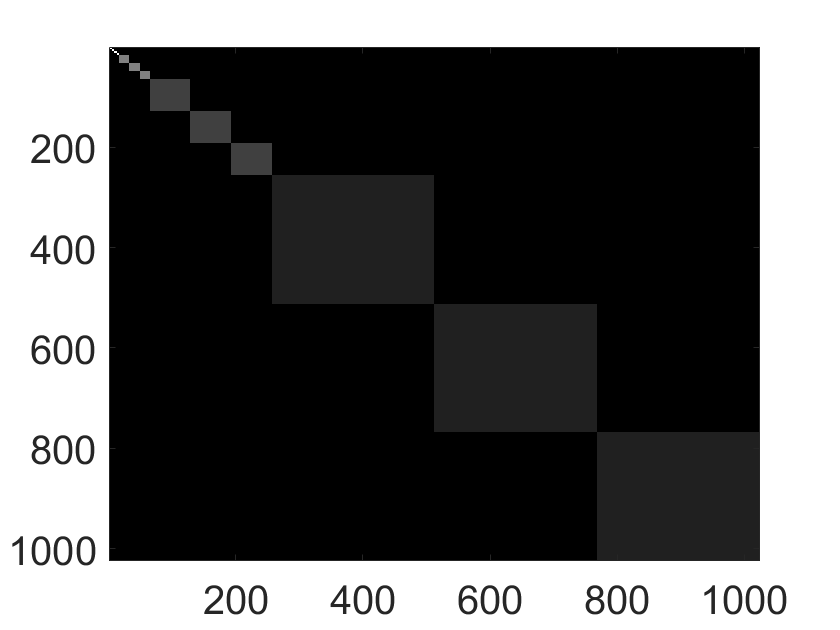}\label{fig:RecHaarWalsh2d}}\quad
	\subfloat[db$2$ - Walsh]{\includegraphics[width=0.31\textwidth]{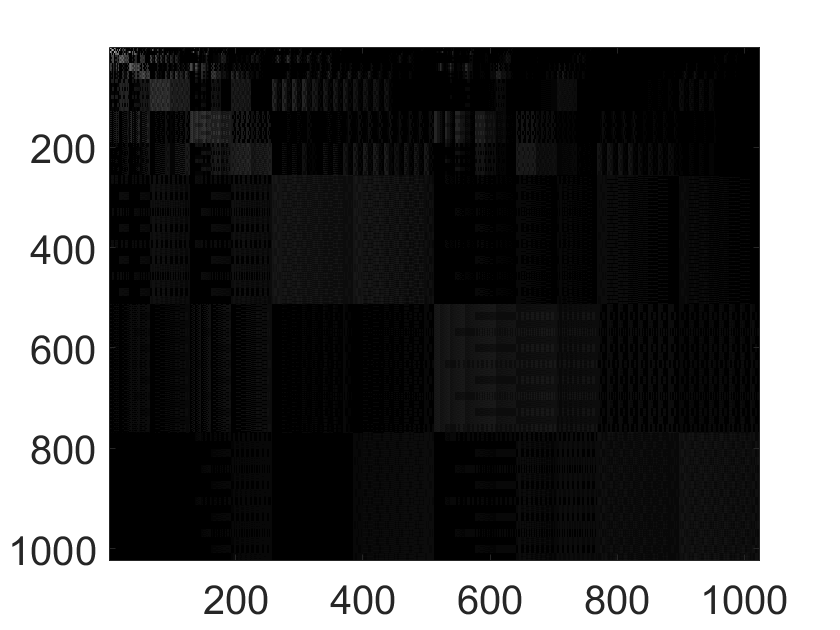}} \quad 
	\subfloat[db$8$ - Walsh]{\includegraphics[width=0.31\textwidth]{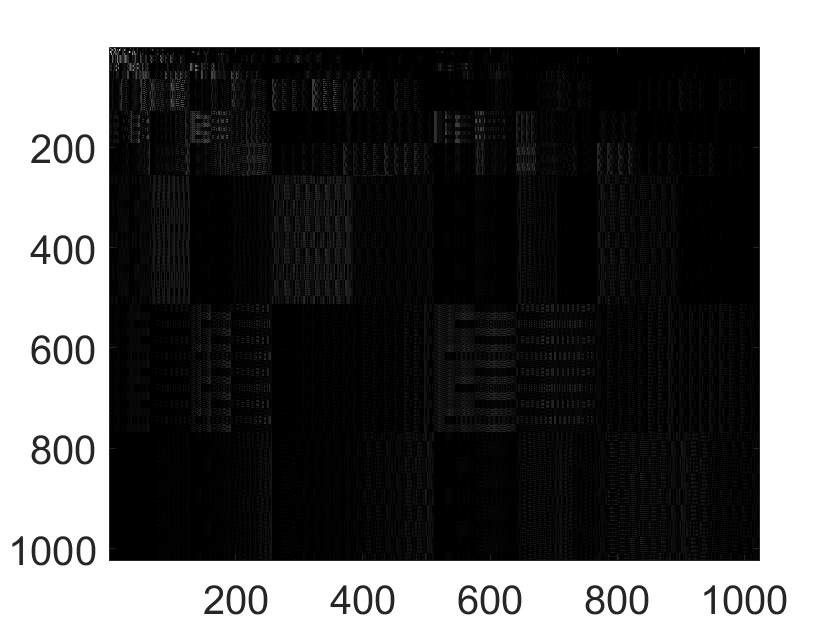}}
	\caption{Reconstruction matrix in two dimensions}\label{fig:RecMatrix2d}
\end{figure}

 After this detailed analysis of the behaviour of the wavelet and the scaling function under the Walsh transform, we are able to proof Theorem \ref{Th:SSRHaar}. 

\begin{proof}[Proof of Theorem \ref{Th:SSRHaar}]
	We want to analyse the subspace angle $\mu(\R,\SM)$ for $N=M$. In detail, we are interested in bounding $\mu(\RM,\SM)$ by  $\theta$ for all $\theta \in (1,\infty)$. Hence, we try to show that $\mu(\mathcal{R}_M,\SM) = 1$ or equally $1/\mu(\mathcal{R}_M,\SM) = 1$  for $M = 2^{dR}$. Due to the fact that the circle $\left\{r \in \RM, ||r||=1 \right\}$ is compact, and the orthogonal projection is continuous there exists $r_0 \in \RM$, $||r_0|| =1$ such that we have
	\begin{equation}
		\frac 1 {\mu(\RM,\SM)} = \inf_{r \in \RM, ||r||=1} ||P_{\SM} r || = ||P_{\SM} r_0 || = 1 - || P_{\SM}^\perp r_0||.
	\end{equation}
	The minimal element $r_0$ can be represented by
	\begin{equation}
	r_0 = \sum_{j=0}^{2^{R}-1} \bigotimes_{i=1}^d \alpha_j \phi_{R,j} \text{ with } \sum_{j=0}^{2^{R}-1} |\alpha_j|^2 =1,
	\end{equation}
	where the multi-index notation is used.
	Then, we have that
	\begin{align}
		P_{\SM}^\perp r_0  
		& = \sum_{n=2^R+1}^{\infty} \langle \sum_{j=0}^{2^{R}-1} \bigotimes_{i=1}^d \alpha_j \phi_{R,j} , \Wal(n,\cdot) \rangle \Wal(n,\cdot)\\
		& = \sum_{n=2^R+1}^{\infty} \sum_{j=0}^{2^{R}-1} \sum_{i=1}^d \alpha_{j_i}  \langle \phi_{R,j_i} , \Wal(n_i, \cdot) \rangle \Wal(n_i, \cdot).
	\end{align}
	With \eqref{Eq:scalingDecay} we get that this sum vanishes. Hence,
	\begin{equation}
		\mu(\RM,\SM) = 1
	\end{equation}
	as desired.
\end{proof}

\subsection{Approximation Rate}

In this chapter we discuss the approximation rate for Walsh functions and wavelets. With approximation theory we get a good insight in the representation properties of bases. For a given orthonormal basis $\left\{r_i\right\}_{i \in \mathbb{N}}$ for $L^2([0,1^d])$ we have that for all $f \in L^2([0,1]^d)$ it holds
$$
f = \sum_{i \in \mathbb{N}} \langle f, r_i \rangle r_i.
$$
Unfortunately, in most applications we cannot store or access $\langle f , r_i \rangle$ for all $i \in \mathbb{N}$ but only for $i =1, \ldots, N$ for some $N \in \mathbb{N}$. Hence, instead of the true function $f$ we can only have an approximation $f_N = \sum_{i=1}^N \langle f, r_i \rangle r_i$. In the field of approximation theory one studies for different representation systems how good this estimate is. Particularly, one is interested in the error
\begin{equation}
\epsilon(N,f) = ||f - f_N||_2^2 = \int|f - f_N|^2 dx = \sum_{i > N} |\langle f, r_i \rangle |^2.
\end{equation}
In general, representation systems are desirable with the property that this error decays very fast in $N$. The reason why is that we get a good representation from only a small amount of information. This results in more efficient compression or less measurement time. 

In \cite{Vegard} it is shown that for all continuous functions in $L^2([0,1]^d)$ the approximation error decays with $\mathcal{O}(N^{-2})$. 
The block artefacts are reflected in the poor approximation rate, and it underlines the need of reconstruction methods, which change the basis to achieve better approximation rates. 

Daubechies wavelets obey this property. This is analysed in detail in \cite{WaveletDef}. Consider the representation with Daubechies wavelets of order $p$, and let the data $f$, which we try to represent, be in the Sobolev space $W^\gamma([0,1]^d)$ for some $\gamma <p$. Then, we get an improved approximation rate of 
\begin{equation}
\epsilon(N,f) = \mathcal{O}(N^{-2\gamma}).
\end{equation}
Hence, for all Daubechies wavelets of order $p \geq 3$ and all functions $f \in W^\gamma([0,1]^d)$ for some $\gamma \geq 2$ we get an improved decay rate with Daubechies wavelets in contrast to the representation with Walsh functions.

From Theorem \ref{Th:SSR} we know that the stable sampling rate is linear for Walsh functions and wavelets and in the following chapter we will even see that the constant is reasonably low, i.e. $S_2 = 2$ for Daubechies $8$ wavelets. This means that with only a constant increase of measurements we get from a decay rate of $\mathcal{O}(N^{-2})$ to an approximation rate of $\mathcal{O}(N^{-2\gamma})$.

\section{Numerical Experiments}\label{Ch:Numerics}

\begin{figure}
	\subfloat[Original Signal]{\includegraphics[width=0.45\textwidth]{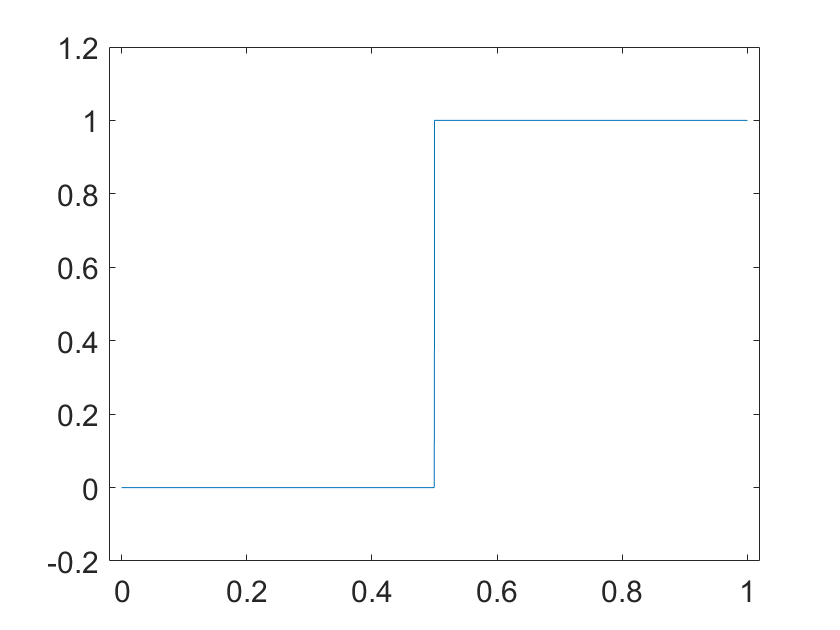}\label{fig:HaarFourOrig}} \quad
	\subfloat[Generalized Sampling with $64$ samples]{\includegraphics[width=0.45\textwidth]{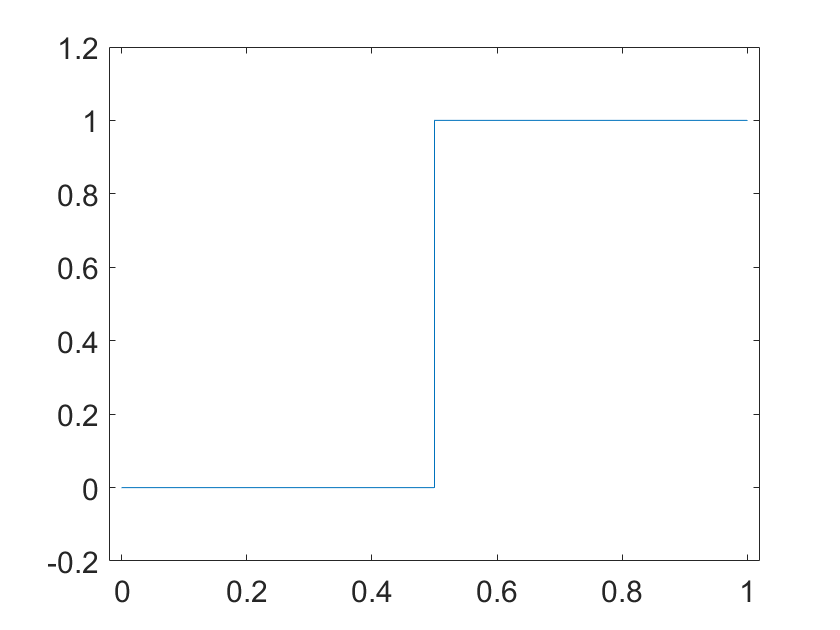}\label{fig:HaarFourGS}} \\
	\subfloat[Truncated Fourier Transform from $64$ samples]{\includegraphics[width=0.45\textwidth]{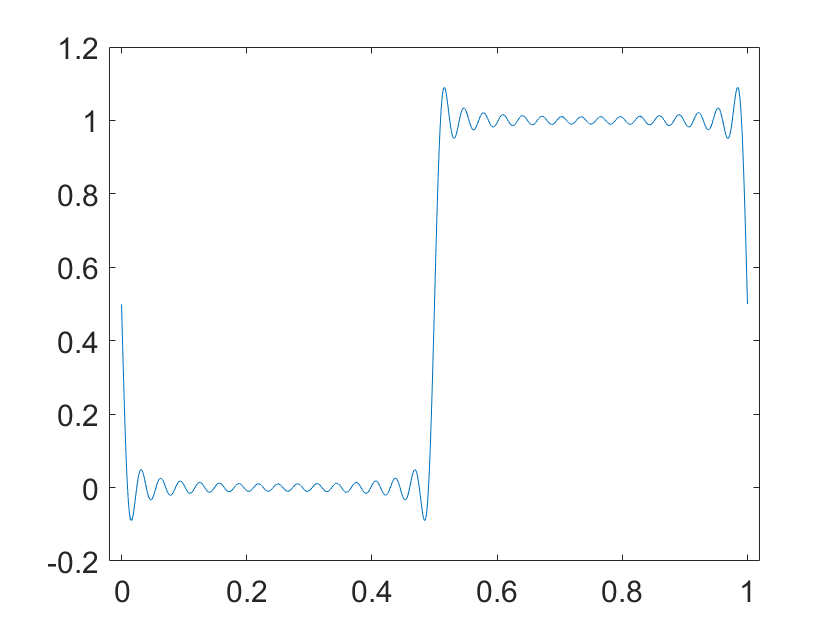}\label{fig:HaarFourTF64}} \quad 
	\subfloat[PBDW-method from $64$ samples]{\includegraphics[width=0.45\textwidth]{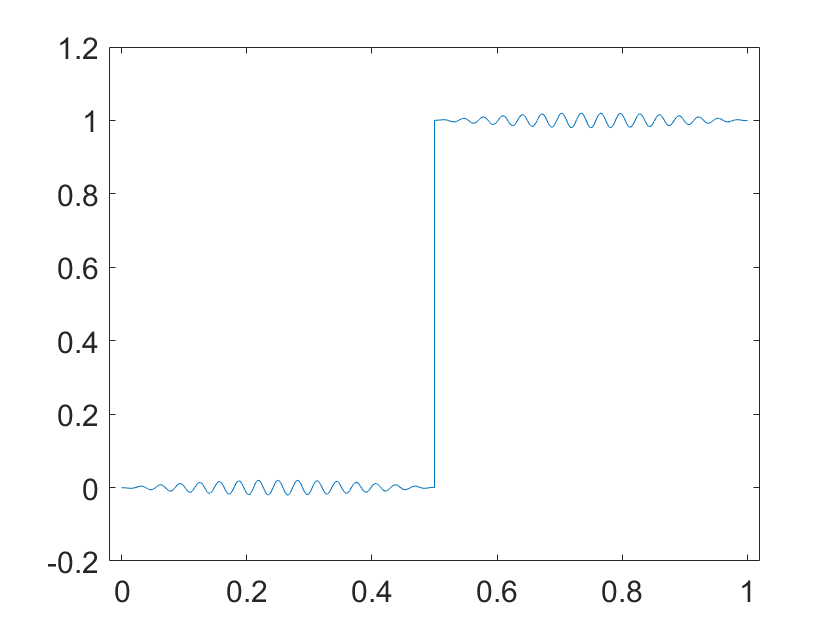}\label{fig:HaarFourDA64}} \\
	\subfloat[Truncated Fourier Transform from $256$ samples]{\includegraphics[width=0.45\textwidth]{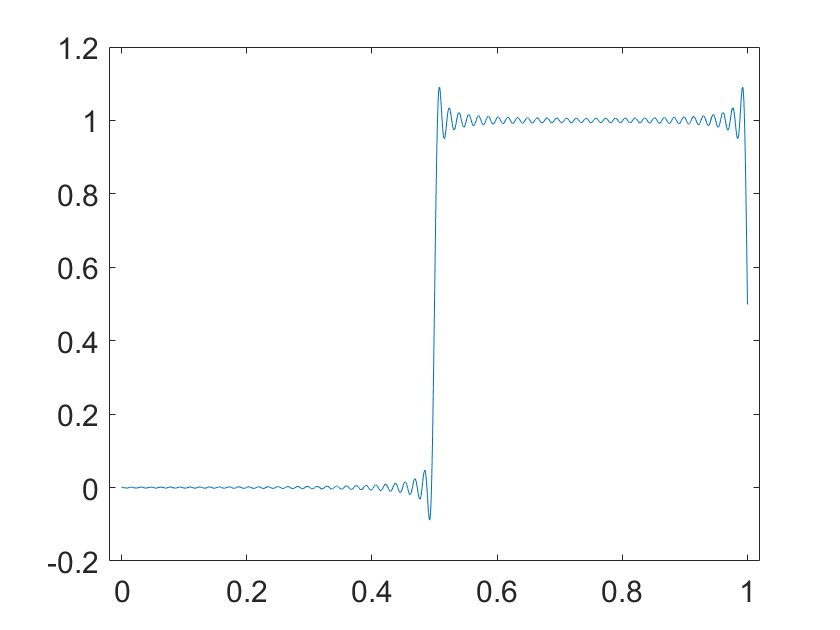}\label{fig:HaarFourTF256}} \quad 
	\subfloat[PBDW-method from $256$ samples]{\includegraphics[width=0.45\textwidth]{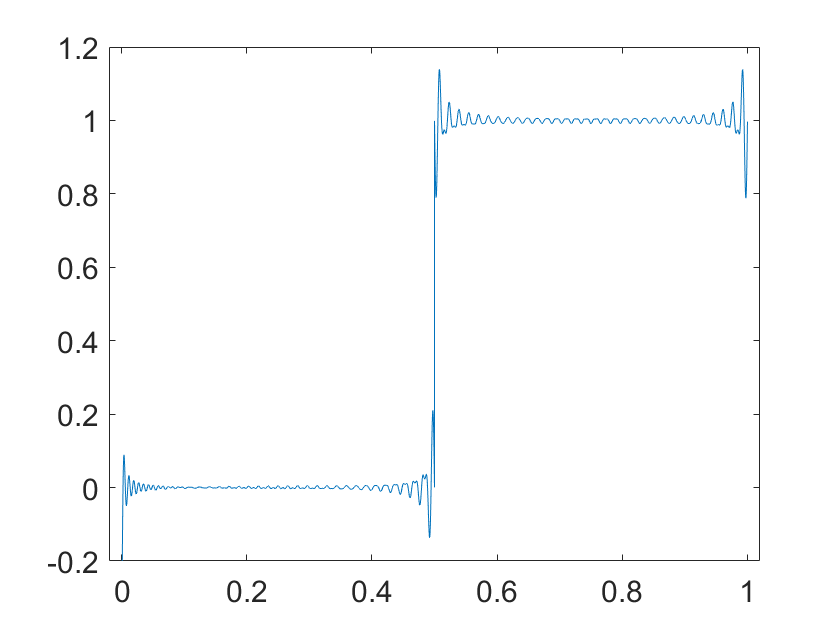}\label{fig:HaarFourDA256}}
	\caption{Reconstruction from Fourier measurements with Haar wavelets and $\operatorname{dim} \R = 32$}
	\label{fig:HaarFour}
\end{figure}

\begin{figure}
	\subfloat[Original Signal]{\includegraphics[width=0.45\textwidth]{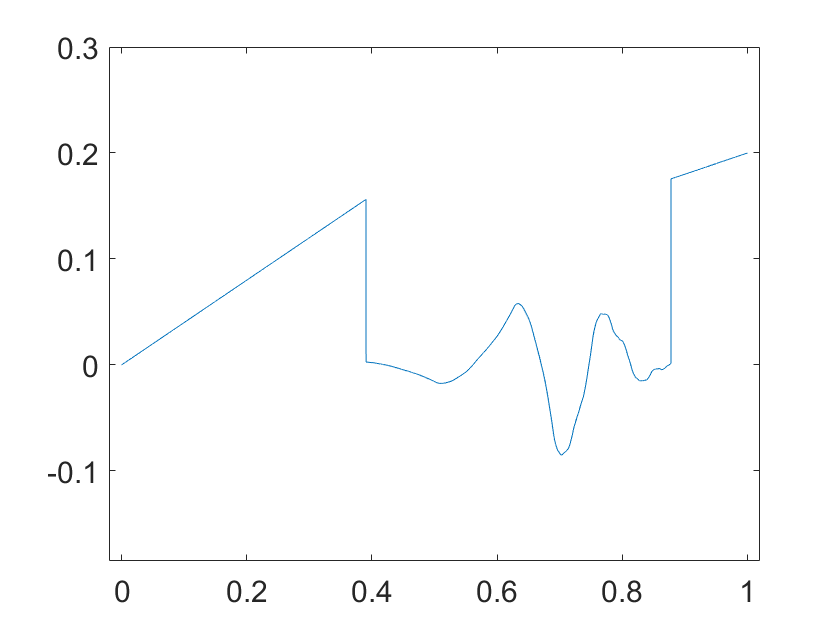}\label{fig:WalHelpOrig}} \quad
	\subfloat[Generalized Sampling with $128$ samples]{\includegraphics[width=0.45\textwidth]{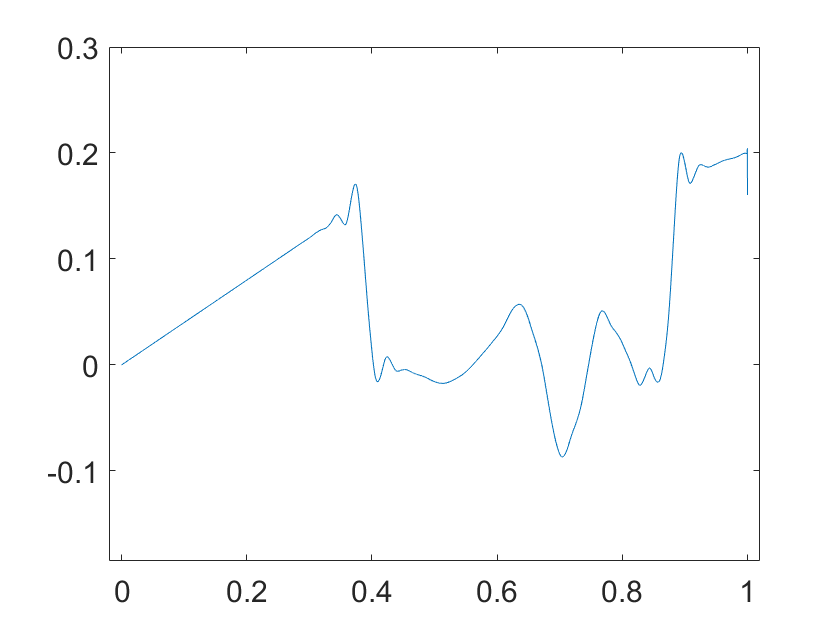}\label{fig:WalHelpGS}} \\
	\subfloat[Truncated Walsh Transform from $128$ samples]{\includegraphics[width=0.45\textwidth]{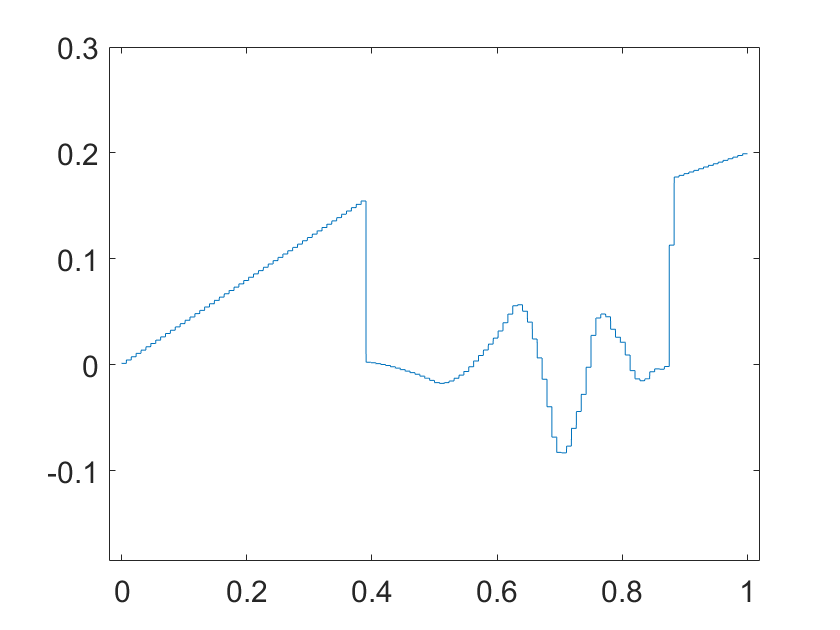}\label{fig:WalHelpTF128}} \quad 
	\subfloat[PBDW-method from $128$ samples]{\includegraphics[width=0.45\textwidth]{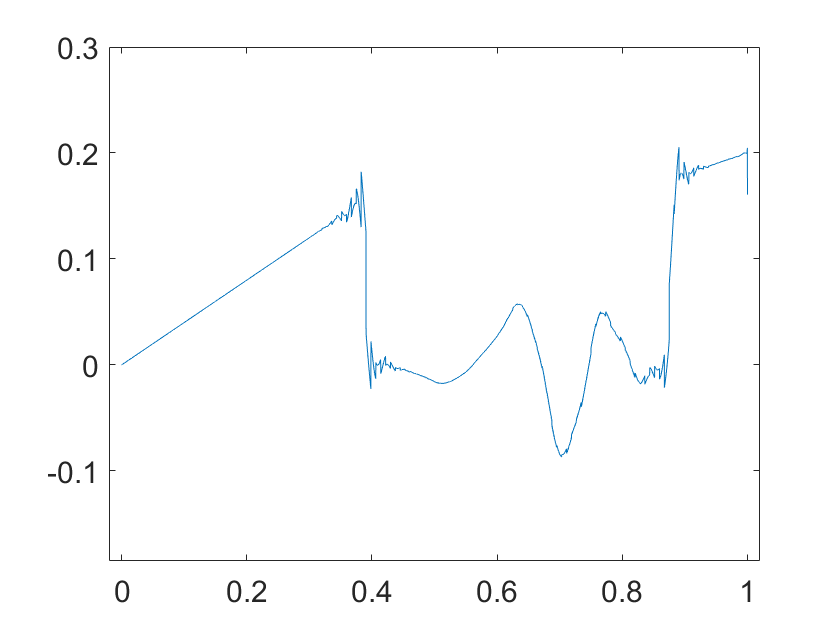}\label{fig:WalHelpDA128}} \\
	\subfloat[Truncated Walsh Transform from $256$ samples]{\includegraphics[width=0.45\textwidth]{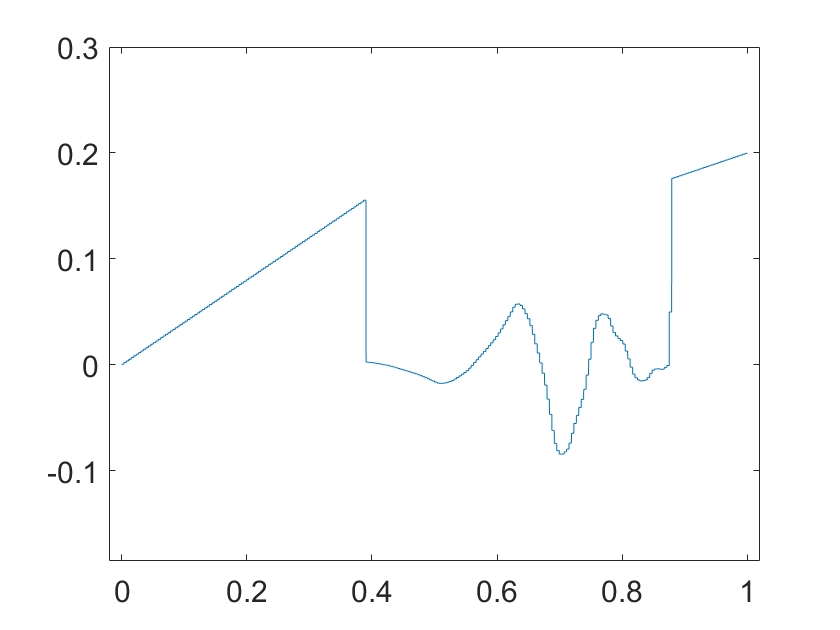}\label{fig:WalHelpTF256}} \quad 
	\subfloat[PBDW-method from $256$ samples]{\includegraphics[width=0.45\textwidth]{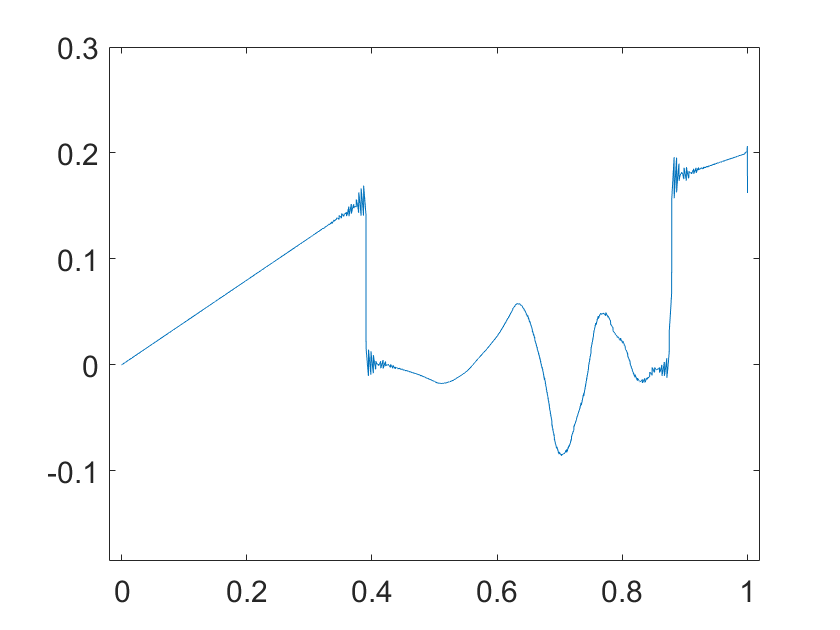}\label{fig:WalHelpDA256}}
	\caption{Reconstruction from binary measurements with Daubechies wavelets $8$ and $\operatorname{dim} \R = 64$}
	\label{fig:WalHelp}
\end{figure}

\begin{figure}
	\subfloat[Original Signal]{\includegraphics[width=0.45\textwidth]{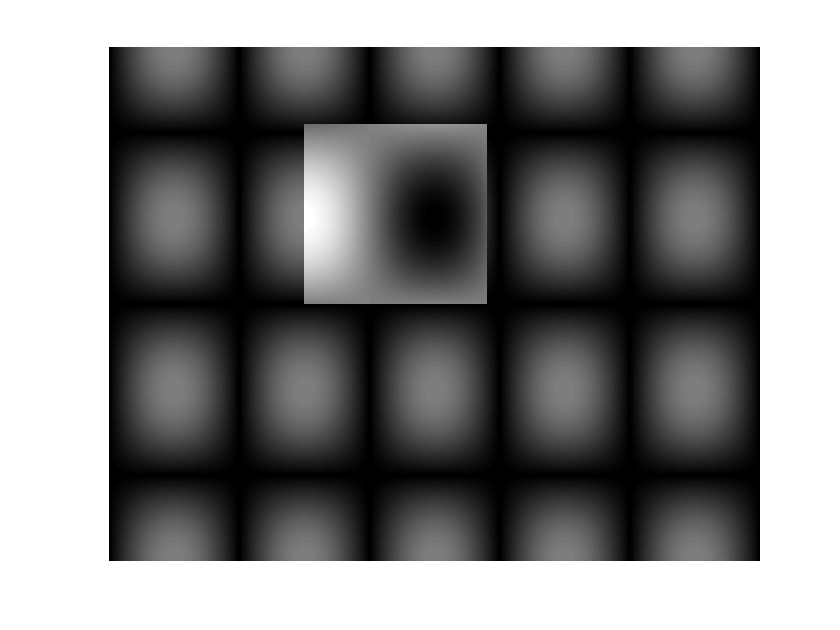}\label{fig:Fourier2DOrig}} \quad
	\subfloat[Truncated Fourier transform with $128^2$ samples]{\includegraphics[width=0.45\textwidth]{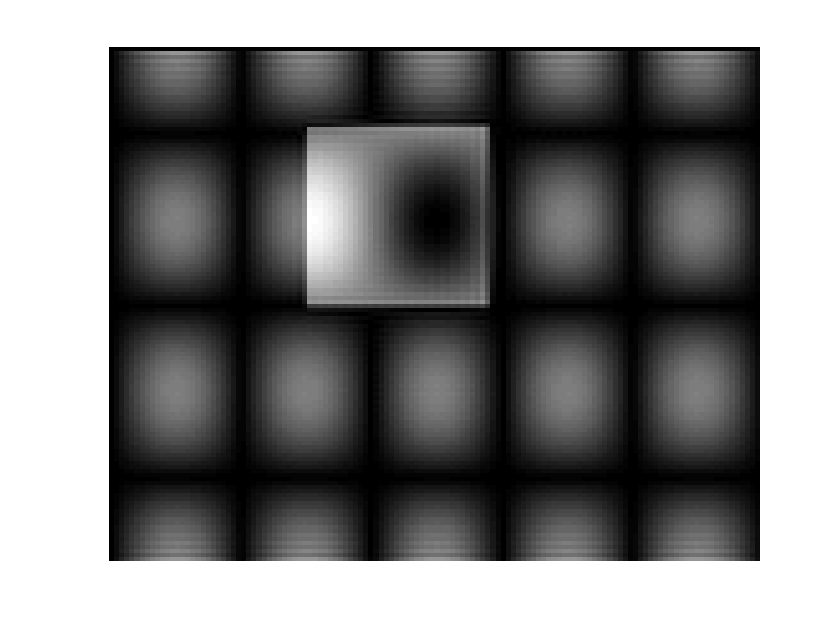}\label{fig:TF2D128}} \\
	\subfloat[Generalized sampling from $128^2$ samples]{\includegraphics[width=0.45\textwidth]{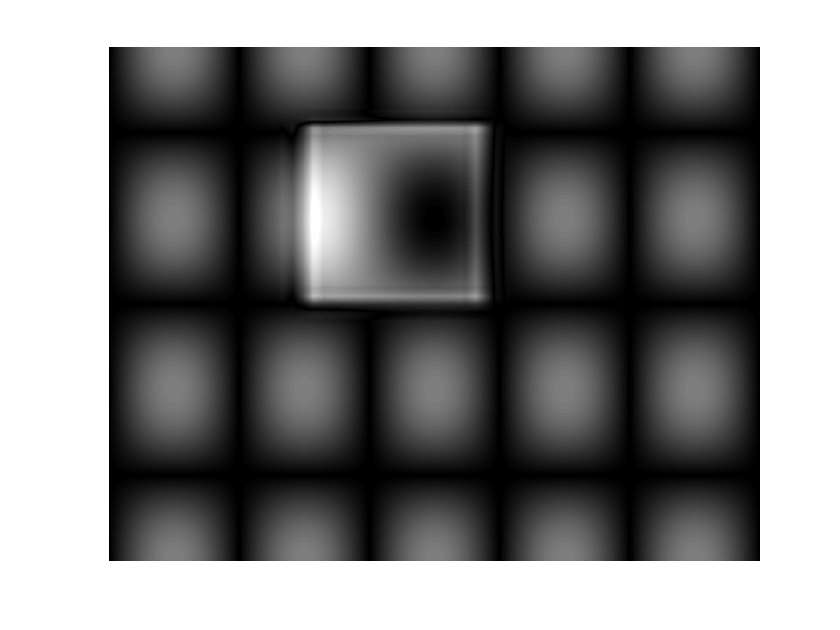}\label{fig:FourierGS2D64128}} \quad 
	\subfloat[PBDW-method from $128^2$ samples]{\includegraphics[width=0.45\textwidth]{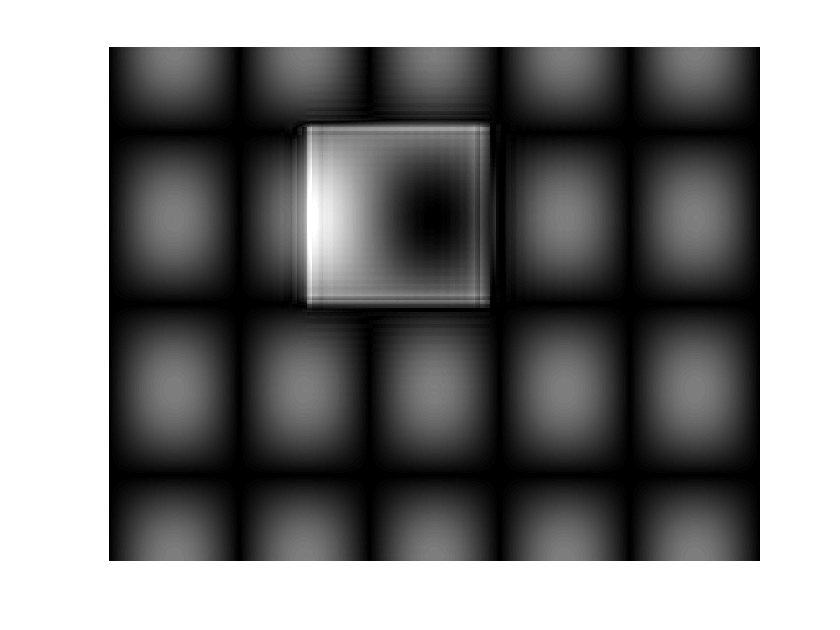}\label{fig:FourierDA2D64128}} \\
	\subfloat[Truncated Fourier Transform from $256^2$ samples]{\includegraphics[width=0.45\textwidth]{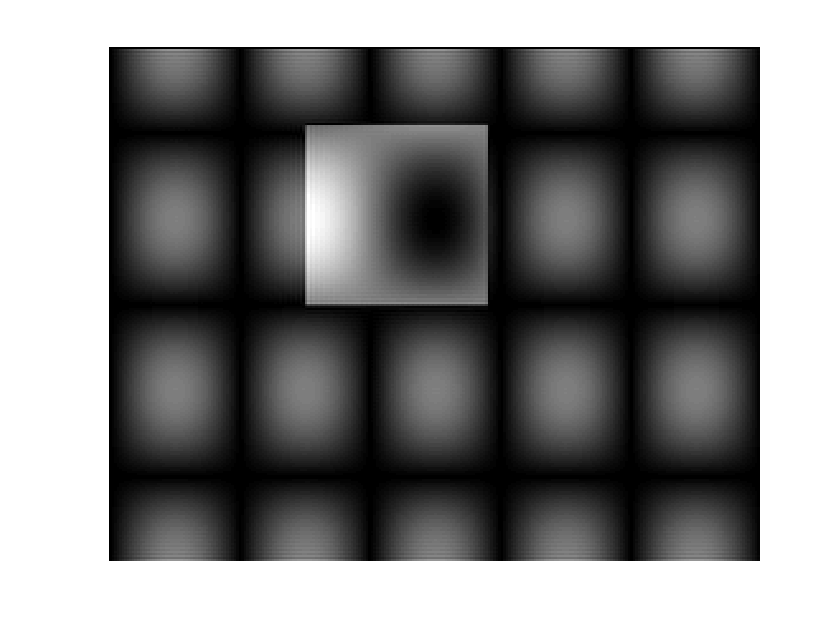}\label{fig:TF2D256}} \quad 
	\subfloat[PBDW-method from $256^2$ samples]{\includegraphics[width=0.45\textwidth]{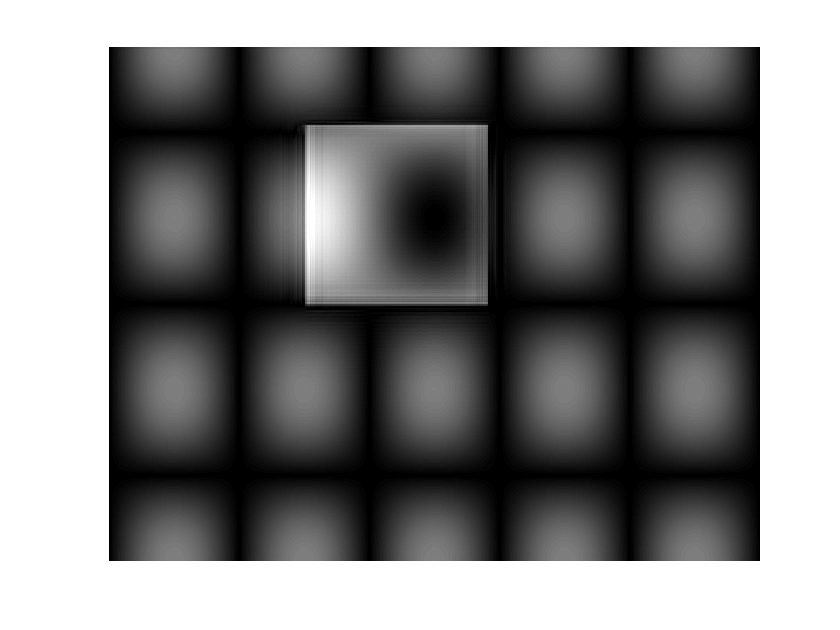}\label{fig:FourierDA2D64256}}
	\caption{Reconstruction from Fourier measurements with Daubechies wavelets $4$ and $\operatorname{dim} \R = 64^2$}
	\label{fig:Fourier2D}
\end{figure}

In this section we want to underline the findings from \S \ref{Ch:Comparison}, especially, the differences between generalized sampling and the PBDW technique with examples using both Fourier and Walsh samples. We emphasise that some of the examples are chosen particularly to highlight the differences between the two methods. Hence, the examples may not reflect typical practical scenarios. In Figure \ref{fig:HaarFour} we look at an extreme case and consider the task of recovering the Haar wavelet from Fourier coefficients. Generalized sampling will obviously recover the function perfectly, given that we choose the Haar basis for the reconstruction space. In this very special case, the recovered solution is consistent with the samples. Nevertheless, this is not true in general and only possible because one gets perfect recovery. In particular, generalized sampling will usually provide non-consistent solutions. The PBDW-method on the other hand is always consistent. The effect of this is that, even when the reconstruction space is fixed, the solution changes with the number of samples. Moreover, in this example it gets closer, as the number of samples increase, to the solution provided by simply truncating the Fourier series.

The next example in Figure \ref{fig:WalHelp} considers Walsh samples. The original function displayed in Figure \ref{fig:WalHelpOrig} is continuous with two jump discontinuities. Moreover, the continuous part is very well represented with Daubechies $8$ wavelets, as demonstrated by the generalized sampling reconstruction in Figure \ref{fig:WalHelpGS}. Nevertheless, there are some artefacts at the jumps. Walsh functions represent the jumps better, but lead to heavy artefacts along the continuous part of the function \ref{fig:HaarFourTF64}. Hence, both spaces have advantages and disadvantages and can represent the signal well in different areas. The PBDW-method allows us to take advantages from both spaces, as it does not force the solution to stay in the reconstruction space. This leads to a different reconstruction quality as in \ref{fig:WalHelpDA128} and \ref{fig:WalHelpDA256}. In this case we also get better results with more samples even though we reconstruct the same number of coefficients.

In the last example we used the code from \cite{MilanaClarice} and consider reconstruction of images from Fourier samples. In Figure \ref{fig:Fourier2D} we see that the PBDW-method provides good results in the $2$D setting and with Fourier measurements. The function is very smooth outside the discontinuous part. Therefore, it can be effectively represented with wavelets in this area. Nevertheless, there are some obvious artefacts around the discontinuity in Figure \ref{fig:FourierGS2D64128}. The artefacts that arise from the truncated Fourier transform are also easy to spot in Figure \ref{fig:TF2D128} and less clear in Figure \ref{fig:TF2D256} due to the increased number of samples. With the PBDW-method it is possible to merge the advantages of both systems and decrease the error, as seen in Figures \ref{fig:FourierDA2D64128} and \ref{fig:FourierDA2D64256}. We can also see that an increased number of samples leads to better performance of the PBDW-method even if the number of reconstructed coefficients stays the same.

\vspace{13pt}
\centerline{ACKNOWLEDGEMENT}
\vspace{13pt}

\noindent LT acknowledges support by the UK Engineering and Physical Sciences Research Council (EPSRC) grant EP/L016516/1 for the University of Cambridge Centre for Doctoral Training, the Cambridge Centre for Analysis. AH acknowledges support from a Royal Society University Research Fellowship.

\renewcommand \thesection{Appendix \Alph{section}:}
\renewcommand \thesubsection{\Alph{section}.\arabic{subsection}}

\end{document}